\documentclass[11pt,reqno]{article}

\usepackage[utf8]{inputenc}

\usepackage[numbers,sort&compress]{natbib}
\usepackage[T1]{fontenc}
\usepackage{microtype}
\usepackage{dsfont}
\usepackage{algorithm}
\usepackage{algorithmic}
\usepackage[bookmarksnumbered,hypertexnames=false,colorlinks=true,linkcolor=blue,urlcolor=blue,citecolor=blue,anchorcolor=green,breaklinks=true,pdfusetitle]{hyperref}
\usepackage{amsmath,amssymb,amsthm,braket,fullpage,enumitem,authblk,mathtools,bbm}
\usepackage{caption,subcaption}
\usepackage[parfill]{parskip}
\usepackage{xcolor}
\usepackage[english]{babel}
\usepackage{tikz}

\newtheorem{thm}{Theorem}[section]
\newtheorem{lem}{Lemma}[section]
\newtheorem{claim}[thm]{Claim}
\newtheorem{definition}[thm]{Definition}

\DeclareMathOperator{\poly}{poly}
\DeclareMathOperator{\polylog}{polylog}

\DeclareMathOperator*{\argmax}{arg\,max}
\DeclarePairedDelimiter{\paren}{\lparen}{\rparen}

\newcommand{\ceil}[1]{\lceil #1 \rceil}
\newcommand{\cost}{\mathrm{T}}
\newcommand{\oracle}{\mathrm{O}}
\newcommand{\unitary}{\mathrm{U}}
\newcommand{\bigO}{\mathcal{O}}

\renewcommand{\braket}[2]{\langle #1|#2\rangle}
\renewcommand{\indent}{\hspace*{5mm}}

\allowdisplaybreaks[4]

\newcommand\blfootnote[1]{%
  \begingroup
  \renewcommand\thefootnote{}\footnote{#1}%
  \addtocounter{footnote}{-1}%
  \endgroup
}

\usetikzlibrary{quantikz}
\tikzstyle{arrow} = [thick,->,>=stealth]
\usetikzlibrary{shapes.geometric, arrows, fit}
\tikzstyle{box_orange} = [rectangle, rounded corners, minimum width=4.5cm, minimum height=1cm,text centered, draw=black,text width=4cm, fill=orange!40]
\tikzstyle{box_yellow} = [rectangle, minimum width=5cm, minimum height=1cm, text centered, text width=4.5cm, draw=black, fill=yellow!30]
\tikzstyle{box_blue} = [rectangle, rounded corners, minimum width=4.5cm, minimum height=1cm,text centered, draw=black, fill=cyan!30,text width=4cm]
\tikzstyle{box_green} = [rectangle, rounded corners, minimum width=4.5cm, minimum height=1cm,text centered, draw=black, fill=green!30,text width=4cm]
\tikzstyle{block} = [rectangle, draw, text width=7em, text centered, rounded corners, minimum height=3em]
\tikzstyle{line} = [draw, -latex]
\usetikzlibrary{patterns}

\usepackage[colorinlistoftodos,size=small]{todonotes}

\begin{document}

\title{Quantum Reinforcement Learning via Policy Iteration}

\date{}

\author[1]{El Amine Cherrat}
\author[1, 2]{Iordanis Kerenidis}
\author[2]{Anupam Prakash}

\affil[1]{Université de Paris, CNRS, IRIF, F-75006, Paris, France}
\affil[2]{QC Ware, Palo Alto, USA and Paris, France}

\maketitle

\begin{abstract}

    Quantum computing has shown the potential to substantially speed up machine learning applications, in particular for supervised and unsupervised learning. Reinforcement learning, on the other hand, has become essential for solving many decision making problems and policy iteration methods remain the foundation of such approaches. 
    In this paper, we provide a general framework for performing quantum reinforcement learning via policy iteration. We validate our framework by designing and analyzing: \emph{quantum policy evaluation} methods for infinite horizon discounted problems by building quantum states that approximately encode the value function of a policy $\pi$; and \emph{quantum policy improvement} methods by post-processing measurement outcomes on these quantum states. Last, we study the theoretical and experimental performance of our quantum algorithms on two environments from OpenAI's Gym.
    
\end{abstract}

\blfootnote{Emails: cherrat@irif.fr, jkeren@irif.fr, anupamprakash1@gmail.com}

\newpage

\section{Introduction}
\label{section:introduction}

    \indent Reinforcement learning has had a great impact in decision making problems, in particular combined with artificial neural networks \citep{Mnih2015HumanlevelCT, Silver2017MasteringTG}. Nevertheless, alternatives to neural networks are still needed for a number of different reasons, first, because the amount of data is expected to continue to grow along with its dimensionality, and, second, neural networks carry vulnerabilities that make them prone to adversarial attacks \cite{Szegedy2014IntriguingPO}. A possible alternative to deep learning for further improving machine learning can be found in quantum computing that has shown to be able to perform tasks beyond the reach of classical computing \cite{Arute2019QuantumSU}. The field of quantum machine learning explores how to design and implement quantum algorithms that could enable machine learning that is faster, more expressive, or more explainable. Using quantum computers, a number of quantum machine learning algorithms have been published for supervised and unsupervised learning \citep{Lloyd2014QuantumPC, Kerenidis2017QuantumRS, Biamonte2017QuantumML, Lloyd2018QuantumGA, Kerenidis2020QuantumAF, Kerenidis2020QuantumGD, Kerenidis2021ClassicalAQ}. Here, we are interested in reinforcement learning and in particular the policy iteration algorithm \cite{Sutton2005ReinforcementLA}:
    
    \begin{center}
        \textit{Policy iteration is an algorithm that, given a Markov Decision Process, \\ generates a sequence of policies converging to the optimal policy in a finite number of steps. }
    \end{center}

    \indent The research on quantum reinforcement learning is so far rather limited. The first approach was developed by \citet{Dong2008QuantumRL} where a quantum environment using a superposition of states and actions is proposed, while temporal difference learning is used for policy evaluation and Grover techniques \cite{Grover1996AFQ} for policy improvement. \citet{Cornelissen2018QuantumGE} proposed a quantum algorithm to evaluate the value function using a phase reward oracle and used quantum gradient estimation \cite{gilyen2019optimizing} to improve the policy. \citet{Ronagh2019QuantumAF} presents a quantum dynamic programming algorithm for solving the finite-horizon processes case. Several works in quantum reinforcement explore the use of variational circuits for value-based and policy-based algorithms \citep{Chen2020VariationalQC, Lockwood2020ReinforcementLW, Skolik2021QuantumAI, Jerbi2021VariationalQP}. More recent work from \citet{Wang2021QuantumAF} combines quantum mean estimation with the quantum maximum searching algorithm to estimate the optimal policy and value function when a generative model for the environment is available. This approach provides a polynomial speedup over the classical value iteration algorithm but does not generalize to the case where such model is not available.
    
    \indent In this work, we define a general framework for quantum reinforcement learning based on quantum policy iteration, by extending the classical approach in \cite{Lagoudakis2003LeastSquaresPI}. We provide algorithms for performing quantum policy iteration and we extend them to the approximate case with linear value functions. Our \textit{quantum policy iteration} algorithms alternate between two steps as their classical counterparts. The \textit{quantum policy evaluation} step uses quantum linear system solvers to produce quantum states that approximately encode the policy value function. The \textit{quantum policy improvement} step improves the actual policy based on measurements performed on these quantum states. In Section \ref{section:quantum-policy-iteration}, we provide an in-depth analysis of our quantum policy iteration algorithm where we prove tight convergence bounds similar to the classical case, and provide efficient methods for building quantum access to the environment and policy parameters. In Section \ref{section:approximate-quantum-policy-iteration}, we generalize our approach to the approximate case with linear value function approximation and provide a model-free implementation. 
    
    \indent The quantum policy iteration methods we develop use quantum linear system solvers for evaluating the policies and hence the running time of these procedures depend explicitly on parameters of the matrices involved in these linear systems (for example the condition number, sparsity or rank) and also on how efficient it is to access these matrices in a quantum way (in other words constructing efficient block encodings). 
    
    \indent In fact, we believe reinforcement learning is an advantageous case for quantum linear algebra precisely due to the character of the linear systems which are usually sparse and well-conditioned. Much of the effort in the paper is to provide explicit constructions of the block encodings for all cases, which enables to bound the parameters in the running time and have a clear idea of when to expect a quantum advantage. For example, we will see in Section \ref{section:quantum-policy-iteration} that for certain environments like the \textsc{FrozenLake}, mazes or other board games, the running time of our quantum method can be thought of as $\bigO(SA + \log(SA)/\epsilon^2)$, where $S,A$ are the states and actions of the game, and $\epsilon$ is the accuracy for retrieving the solution from the quantum linear system solver, which one can compare with the $\bigO((SA)^\omega)$ running time of the classical linear system method. We also provide a similar comparison between running times of classical and quantum approximate policy iteration methods in Section \ref{section:approximate-quantum-policy-iteration}. Last, in Section \ref{section:applications}, we simulate our quantum algorithms for the \textsc{FrozenLake} and \textsc{InvertedPendulum} environments to show that our quantum algorithms converge well and can be considerably faster in practice, and we also describe how to build the necessary block encodings.
    
    \indent Overall, our methodology provides a general framework for infinite-horizon problems in the model-based and model-free case, and it encompasses many different ways of performing policy evaluation and improvement, including deep learning techniques, thus, enabling theoretical analysis of quantum reinforcement learning.

\section{Preliminaries}
\label{section:preliminaries}

\subsection{Reinforcement learning}
\label{subsection:reinforcement-learning}
    
    \indent The aim of reinforcement learning\footnote{For a more detailed introduction to reinforcement learning, we recommend \cite{Sutton2005ReinforcementLA}. } is to train an agent to discover the policy that maximizes the agent's performance in terms of the discounted future reward,  while interacting with the environment, receiving only a reward signal. The agent can take actions in a set of possible actions based on a policy that maps each state with actions to take. This interaction is summarized in Figure \ref{fig:env}. 
    
    \indent More formally, we consider the infinite-horizon discounted decision problem with a \emph{state set} $\mathcal{S}$ and a finite \emph{action set} $\mathcal{A}$. At each time-step $t$, the agent receives a representation of the environment's \emph{state} $s_t \in \mathcal{S}$,  selects an \emph{action} $a_t \in \mathcal{A}$ and receives a \emph{reward} $r_t \in [0,1]$. Denoting by $p(s,a,s')$ the probability that $s'$ will occur and by $r(s,a)$ the average reward perceived after taking action $a$ whilst in state $s$, the usual framework used to describe the environment's elements in reinforcement learning are \emph{Markov Decision Processes} (MDP) which are fully defined by giving tuples of the form $\mathcal{M}=(\mathcal{S},\mathcal{A},P,R,\gamma)$ where $P=[p(s,a,s')]_{s,a,s'}\in\mathbb{R}^{S\times A\times S}$ is the transition matrix, $R=[r(s,a)]_{s,a}\in\mathbb{R}^{S\times A}$ is the reward vector and $\gamma\in(0,1)$ is the discount factor. For the rest of the paper, we will denote by $S$ the size of the state space $\mathcal{S}$, by $A$ the size of the action space $\mathcal{A}$ and by $\Gamma$ the effective time horizon: 
    $$ \Gamma = \frac{1}{1-\gamma}$$
    
\newpage
\begin{figure}[t]
    \centering
        \begin{tikzpicture}[node distance = 6em, auto, thick]
            \node [block] (Agent) {\textbf{Agent}};
            \node [block, below of=Agent] (Environment) {\textbf{Environment}};
             \path [line] (Agent.0) --++ (4em,0em) |- node [near start]{$a_t$} (Environment.0);
             \path [line] (Environment.190) --++ (-6em,0em) |- node [near start] {$s_t$} (Agent.170);
             \path [line,dotted] (Environment.170) --++ (-4.25em,0em) |- node [near start, right] {$r_t|s_{t+1}$} (Agent.190);
        \end{tikzpicture}
    \caption{The agent–environment interaction \cite{Sutton2005ReinforcementLA}.}
    \label{fig:env}
\end{figure}
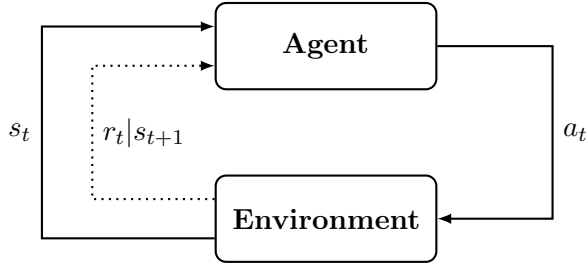

    \indent The agent's behavior is modeled by some policy $\pi$, where $\pi(s,a)$ represents the probability of selecting action $a$ given the state $s$. Moreover, for every policy $\pi$, we define its value function $Q^\pi\in\mathbb{R}^{S\times A}$ as: $$ Q^\pi(s,a) = \mathbb{E}\Big[\sum_{t=0}^{\infty} \gamma^t r(s_t,a_t)\Big|s_0=s, a_0=a, a_t \sim \pi(s_t), s_{t+1}\sim p(s_t,a_t)\Big] $$ which denotes the cumulative reward received by the agent when starting from $(s,a)$ and playing according to $\pi$. The value function $Q^\pi$ is the unique solution of the \emph{Bellman equation} $\mathcal{T}^\pi Q^\pi = Q^\pi$ where $\mathcal{T}^\pi$ is the operator acting on $\mathbb{R}^{S\times A}$ such that: $$\mathcal{T}^\pi Q(s,a)  =  r(s,a) + \gamma \sum_{s'a'} p(s,a,s') \pi(s',a') Q(s',a')$$ 
    
    \indent The goal in reinforcement learning is to find the optimal value function $Q^*=\sup_\pi Q^\pi$ over all  policies $\pi$ that maximizes the value function for all state-action pairs. A classical result about MDPs is the existence of a policy, referred to as the optimal policy $\pi^*$ that reaches these optimal values such that $Q^{\pi^*}= Q^*$ and verifies $\pi^*(s) \in \argmax \{Q^*(s,a) | a\in \mathcal{A}\}$ for every $s\in\mathcal{S}$.
    
\subsection{Policy iteration}
\label{subsection:policy-iteration}
    
    \indent The idea of \emph{Policy Iteration} (PI) is to build a sequence of deterministic policies $\{\pi_t\}_{t\in\mathbb{N}}$ that converges to the optimal policy $\pi^*$ (Figure \ref{fig:scheme-policy-iteration}). The algorithm is initialized with a random policy $\pi_0$ and iteratively alternates between two phases. The first phase, called \emph{policy evaluation}, computes the value function $Q^{\pi}$ of the actual policy $\pi$ while the second phase, called \emph{policy improvement}, uses this value function to output an improved policy $\pi'$, usually using a greedy approach with respect to the current value function. 
    
    \indent A generalization of policy iteration is  \emph{Approximate Policy Iteration} (API) algorithms which are used when the environment model is unknown or the state and action spaces are large. Exact representations of the value function $Q^\pi$ and the policy $\pi$ can be replaced by adjustable parameters and many approximation algorithms \citep{Bertsekas2011ApproximatePI,Scherrer2012ApproximateMP} follow the scheme in Figure \ref{fig:scheme-approximate-policy-iteration}. Next, we present a classical result in reinforcement learning that guarantees the convergence of approximate policy iteration:
    
\newpage
\begin{figure}[t]
    \begin{center}
        \begin{subfigure}[b]{0.45\linewidth}
            \centering
            \resizebox{\linewidth}{!}{
                \begin{tikzpicture}[node distance=2cm,scale=0.5]
                \node[font=\bfseries] (function) [box_yellow] {Value Function $Q^\pi$};
                \node[font=\bfseries] (pi_improv) [box_orange, below of=function, ,yshift=-0.8cm, xshift=-3cm] {Policy \\ Improvement};
                \node[font=\bfseries] (pi_eval) [box_orange, right of=pi_improv, xshift=4cm] {Policy \\ Evaluation};
                \node[font=\bfseries] (policy) [box_yellow, below of=pi_improv, xshift=3cm, yshift=-0.8cm] {Policy $\pi$};
                \draw [arrow] (pi_eval) |- (function);
                \draw [arrow] (function) -| (pi_improv);
                \draw [arrow] (pi_improv) |- (policy);
                \draw [arrow] (policy) -| (pi_eval);
                \end{tikzpicture}
            }
            \caption{Policy Iteration}
            \label{fig:scheme-policy-iteration}
        \end{subfigure}
        \begin{subfigure}[b]{0.44\textwidth}
        \centering
            \resizebox{\linewidth}{!}{
                \begin{tikzpicture}[node distance=2cm]
                \node[font=\bfseries] (function) [box_yellow] {Approximate \\ Value Function $\widetilde{Q}^\pi$};
                \node[font=\bfseries] (pi_improv) [box_orange, below of=function, xshift=-3cm] {Policy \\ Improvement};
                \node[font=\bfseries] (pi_actor) [box_blue, below of=pi_improv] {Policy \\ Projection};
                \node[font=\bfseries] (pi_eval) [box_orange, right of=pi_actor, xshift=4cm] {Policy \\ Evaluation};
                \node[font=\bfseries] (pi_critic) [box_blue, right of=pi_improv, xshift=4cm] {Value Function \\ Projection};
                \node[font=\bfseries] (policy) [box_yellow, below of=pi_actor, xshift=3cm] {Approximate Policy $\tilde{\pi}$};
                \node[font=\bfseries][draw, thick, rounded corners, dashed, inner xsep=0.5em, inner ysep=0.5em, fit= (pi_eval) (pi_critic)] (box_eval) {};
                \node[font=\bfseries][draw, thick, rounded corners, dashed, inner xsep=0.5em, inner ysep=0.5em, fit= (pi_improv) (pi_actor)] (box_improv) {};
                \draw [arrow] (box_eval) |- (function);
                \draw [arrow] (function) -| (box_improv);
                \draw [arrow] (box_improv) |- (policy);
                \draw [arrow] (policy) -| (box_eval);
                \draw[thick] (pi_actor) -- (pi_improv);
                \draw[thick] (pi_critic) -- (pi_eval);
                \end{tikzpicture}
            }
            \caption{Approximate Policy Iteration }   
            \label{fig:scheme-approximate-policy-iteration}
        \end{subfigure}
        \caption{General schemes of PI (\ref{fig:scheme-policy-iteration}) and API (\ref{fig:scheme-approximate-policy-iteration}) as reproduced in \cite{Lagoudakis2003LeastSquaresPI}. }
        \label{fig:schemes-classical}
    \end{center}
\end{figure}
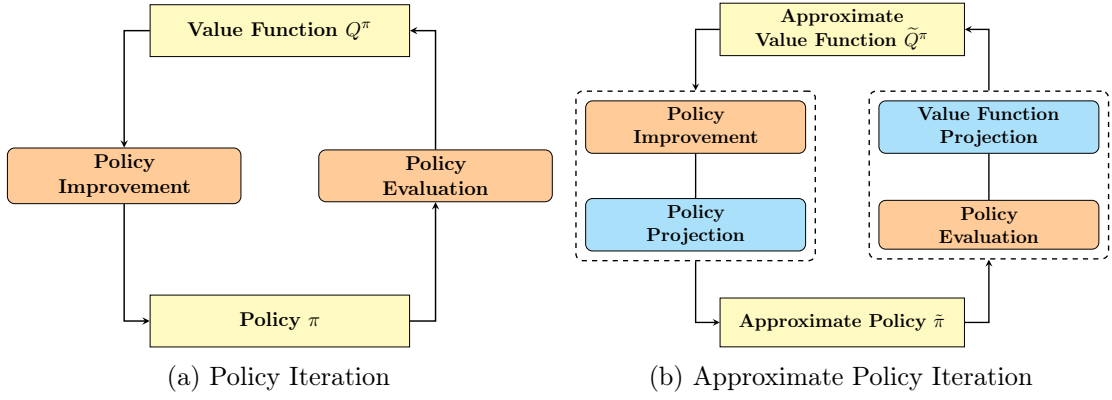

    \begin{thm}[Error bound of API \citep{Bertsekas2019Reinforcement, Lagoudakis2003LeastSquaresPI}]
    \label{theorem:bound-approximate-policy-iteration}
        Let $\{\pi_t\}_{t\in\mathbb{N}}$ be the sequence of policies generated by an approximate policy iteration with a greedy update and let $\{\widehat{Q}^{\pi_t}\}_{t\in\mathbb{N}}$ be the corresponding approximate value functions. Then, this sequence satisfies the following suboptimality bound: $$ \limsup\limits_{t\xrightarrow{}+\infty} \|Q^*-Q^{\pi_t}\|_\infty\leq 2\gamma \Gamma^2\limsup\limits_{t\xrightarrow{}+\infty}\|\widehat{Q}^{\pi_t}-Q^{\pi_t}\|_\infty $$
    \end{thm}

\subsection{Quantum computing}
\label{subsection:quantum-computing}

    \indent Quantum computing\footnote{For a detailed introduction to quantum computing, we recommend \cite{Nielsen2000QuantumCA}.} is a new paradigm for computing that uses the postulates of quantum mechanics in order to encode and compute with information. While classical systems can be only in one state at a time, namely a bit can be either in state $0$ or $1$, quantum systems can be in a superposition of multiple states at the same time, namely a \textit{qubit}, which is the carrier of quantum information, can be in a superposition of the states $\ket{0}$ and $\ket{1}$, i.e.\ it can be written as $\ket{x} = \alpha \ket{0} + \beta\ket{1}$. The qubit $\ket{x}$ corresponds to a unit vector of the Hilbert space $\mathcal{H}=\text{span}\{\ket{0},\ket{1}\}$ with $\alpha,\beta \in \mathbb{C}$ and $|\alpha|^2+|\beta|^2=1$.
    
    \indent The qubit can be generalized to $n$-qubit states, which are unit vectors of $\mathcal{H}_n=\otimes^n \mathcal{H}\simeq\mathbb{C}^{2^n}$. Denoting by $[2^n]$ the set $\{0,\dots,2^n-1\}$ and by $\{\ket{i}\}_{i\in[2^n]}$ the computational basis of $\mathcal{H}_n$, an $n$-qubit state can be written as $ \ket{x} = \sum_{i\in[2^n]} \alpha_i \ket{i} $ with $\alpha_i=\braket{x}{i} \in \mathbb{C}$ and $\sum_{i\in[2^n]} |\alpha_i|^2=1$. Quantum states evolve by applying unitary operators on them, namely applying a unitary operator $\unitary$ (a $2^n \times 2^n$ unitary matrix) on an $n$-qubit state $\ket{x}$ results in the quantum state $\ket{\unitary x}$. In addition, quantum states can be measured and the probability that the measurement of the state $\ket{x}$ gives outcome $i$ is $|\alpha_i|^2$. Note that this is a probability distribution over $[2^n]$. 
    
    \indent The quantum state corresponding to a vector $\mathbf{b}=[b_i]_{i\in[n]}\in\mathbb{R}^n$ is defined as the $\ceil{\log n}$-qubit state $\ket{\textbf{b}}\in\mathbb{C}^{n}$ such that: $$ \ket{\mathbf{b}} = \frac{1}{\|\mathbf{b}\|}\sum_{i\in[n]} b_i \ket{i} $$ where $\ket{i}$ represents the vector $e_i$ the $i$-th vector of the standard basis of $\mathbb{R}^n$. Next, we define the notion of quantum access to a matrix as used in the block-encoding framework \cite{Chakraborty2019ThePO}.
    
    \begin{definition}[Matrix block-encoding \cite{Chakraborty2019ThePO}]
    \label{definition:block-encoding}
        Let $r \in \mathbb{N}$ and $\mu \in \mathbb{R}_+$. We say that we have a $\mu$-block-encoding $\unitary_\mathbf{A}$ of a symmetric matrix $\mathbf{A} \in \mathbb{R}^{n\times n}$ with $\|\mathbf{A}\|\leq \mu$, in total cost $\cost_\mathbf{A} = \cost + r$ if there exists a unitary $\unitary_\mathbf{A}$ acting on $\ceil{\log n}+r$ qubits that can be implemented using $\cost$ elementary gates such that: 
        $$\unitary_\mathbf{A}  = \begin{pmatrix} \mathbf{A}/\mu & . \\ . & . \end{pmatrix}$$
    \end{definition}
    
    \indent This framework, introduced in \citep{Chakraborty2019ThePO,Gilyn2019QuantumSV}, represents any sub-normalized matrix as the top-left block of a unitary $\unitary_\mathbf{A}$. This definition generalizes to any matrix $\mathbf{A}\in\mathbb{R}^{m\times n}$ by building the block-encoding of its symmetrized version $\overline{\mathbf{A}}\in\mathbb{R}^{(m+n)\times(m+n)}$ such that: 
    $$\overline{\mathbf{A}}  = \begin{pmatrix} 0 & \mathbf{A} \\ \mathbf{A}^\dag & 0 \end{pmatrix}$$
    
    \indent The property of constructing a block-encoding of a matrix $\mathbf{A}=[a_{ij}]_{i,j} \in \mathbb{R}^{m\times n}$ can be reduced to being able to perform certain mappings regarding the rows of the matrix $\mathbf{A}^{(p)}$ and the columns of the matrix $\mathbf{A}^{(1-p)}$ for some $p\in[0,1]$ where $\mathbf{A}^{(k)}$ denotes the matrix with elements $(a_{ij})^k\in\mathbb{C}$.
    If we define ${\rm s}_q(\mathbf{A})=\max_{i\in[m]} \|a_i\|_q^q$ to be the maximum $\ell_q$-norm over the rows $a_i$ of $\mathbf{A}$, then we have the following result: 
    
    \begin{lem}[Constructing block-encodings \cite{Chakraborty2019ThePO,Gilyn2019QuantumSV}]
    \label{lemma:constructing-blocks} 
        Let $p\in[0,1]$ and $\alpha,\beta\in\mathbb{R}_+$. Let $\mathbf{A}\in \mathbb{R}^{m\times n}$ be a matrix with ${\rm s}_{2p}(\mathbf{A})\leq \alpha^2$ and  ${\rm s}_{2(1-p)}(\mathbf{A}^\top)\leq \beta^2$. We can implement an $\alpha\beta$-block-encoding of $\mathbf{A}$ by applying a constant number of times the unitaries $\oracle^r$ and $\oracle^c$ such that: 
        \begin{equation*}
            \centering
            \begin{aligned}
                \oracle^r:&  \ket{i,0} \longrightarrow \frac{1}{\alpha}\Big[\sum_{j\in[m]}a_{ij}^{p}\ket{i,j}\Big] + \ket{G^\perp_i}\\
                \oracle^c:&  
                \ket{0,j} \longrightarrow  \frac{1}{\beta} \Big[\sum_{i\in[n]} a_{ij}^{1-p}\ket{i,j}\Big]  + \ket{G^\perp_j} \\
            \end{aligned}
        \end{equation*}
        where $\ket{G_k^\perp}$ denotes some unnormalized garbage quantum state such that $\braket{G_k^\perp}{i,j}=0$ for all $i,j$.
    \end{lem}
    
    \indent Next, we introduce different algorithms and results for quantum linear algebra using block-encodings. The first result (Theorem \ref{theorem:quantum-arithmetics}) describes techniques for performing matrix arithmetics in this framework. The second algorithm is the \emph{quantum linear system solver} whose running time depends on the quantity $\mu$ and the \emph{condition number} of the matrix $\mathbf{A}$ defined as:
    $$ \kappa = \frac{\max_{\mathbf{b}\neq0} \{\|\mathbf{A}\mathbf{b}\|/\|\mathbf{b}\|\}}{\min_{\mathbf{b}\neq0} \{\|\mathbf{A}\mathbf{b}\|/\|\mathbf{b}\|\}}$$
    
    \indent The general idea of solving linear algebra problems with quantum computing is based on the singular value decomposition of matrices. This decomposition is a generalization of eigendecomposition of a positive semidefinite normal matrix and can be used to accelerate algebra and optimization procedures. We state the cost guarantees for the state-of-the-art linear algebra procedures using block-encodings \cite{Chakraborty2019ThePO} (Theorem \ref{theorem:quantum-linear-algebra}). Another quantum algorithm we use is a way to recover efficiently a classical approximation to any quantum state in the $\ell_\infty$-norm \cite{Kerenidis2020QuantumAF} (Theorem \ref{theorem:tomography}).
    
    \begin{thm}[Matrix arithmetics with block-encodings \citep{Chakraborty2019ThePO,Gilyn2019QuantumSV}]
    \label{theorem:quantum-arithmetics}
        Suppose that we have a $\mu_i$-block-encoding $\unitary_i$ of the matrix $\mathbf{A}_i$ at a cost $\cost_i$ for all $i\in[m]$. Then we can implement with cost $\bigO(\sum_{i\in[m]} \cost_i)$ a $(\prod_{i\in[m]} \mu_i)$-block-encoding of $\prod_{i\in[m]} \mathbf{A}_i$ and a $(\sum_{i\in[m]}  |\lambda_i| \mu_i)$-block-encoding of $\sum_{i\in[m]} \lambda_i \mathbf{A}_i$.
    \end{thm}
    
    \begin{thm}[Linear algebra with block-encodings \cite{Chakraborty2019ThePO}]
    \label{theorem:quantum-linear-algebra}
        Let $\mathbf{A} \in \mathbb{R}^{n\times n}$ be a matrix such that $\|\mathbf{A}\|=1$ and let $\epsilon>0$ be the precision parameter. Given a $\mu$-block-encoding of $\mathbf{A}$ with cost $\cost_\mathbf{A}$ and a procedure preparing a state $\ket{\mathbf{b}}$ with cost $\cost_\mathbf{b}$,  then there exists quantum algorithms such that with probability at least $\paren{1-1/\poly\paren{n}}$ return a state $\ket{\mathbf{x}}$ such that $\|\ket{\mathbf{x}}-\ket{\mathcal{A}\mathbf{b}}\|\leq\epsilon$ for $\mathcal{A}\in\{\mathbf{A},\mathbf{A}^{-1},\mathbf{A}^\top\}$ with cost\footnote{If $\|\mathbf{A}\|\neq1$, then we rescale the factor $\mu$ to ${\mu}/{\|\mathbf{A}\|}$.} $\bigO\paren{ \kappa\paren{\mu \cost_\mathbf{A}+\cost_\mathbf{b}}\polylog\paren{\kappa/\epsilon}}$.
    \end{thm}
    
    \begin{thm}[Vector tomography \cite{Kerenidis2020QuantumAF}]
    \label{theorem:tomography}
        Let $\mathbf{x}\in\mathbb{R}^n$ be a normalized vector. Given a procedure preparing $\ket{\mathbf{x}}$ with cost $\cost_\mathbf{x}$, there is a tomography algorithm that with probability at least $(1-1/\poly(n))$ produces a unit vector $\widetilde{\mathbf{x}}\in\mathbb{R}^n$ 
        such that $\|\mathbf{x}-\widetilde{\mathbf{x}}\|_\infty\leq \epsilon$ with cost $\bigO(\cost_\mathbf{x} \log(n) / \epsilon^2 )$.
    \end{thm} 
    
    \indent Moreover, we also introduce this claim from \cite{Kerenidis2020QuantumGD} that bounds the distance between two quantum states in terms of the distance between the corresponding unnormalized vectors:
    
    \begin{claim}
    \label{claim:angle-error}
        Let $\theta$ be the angle between two vectors $\mathbf{x}$ and $\mathbf{y}$ and assume that $\theta < \pi/2$. Then $\|\mathbf{x}-\mathbf{y}\|\leq \epsilon$ implies $\|\ket{\mathbf{x}}-\ket{\mathbf{y}}\|\leq\sqrt{2}\epsilon/\|\mathbf{x}\|$.
    \end{claim}

\section{Quantum policy iteration}
\label{section:quantum-policy-iteration}

\subsection{General framework for quantum policy iteration}
\label{subsection:qpi-framework}

    \indent We start by providing a general framework for quantum reinforcement learning by appropriately extending the policy iteration scheme to the quantum case. In this section we look at the case where we work directly with exact representations of the value function $Q^{\pi}$, while in Section \ref{section:approximate-quantum-policy-iteration}, we generalize the algorithm to the approximate case.
    Similarly to the classical case, we want to build quantum methods that generate a sequence of policies improving at each iteration and converging to an approximation of the optimal policy $\pi^*$. We refer to our framework as \textit{Quantum Policy Iteration} (QPI) and we summarize the general procedure in Figure \ref{fig:scheme-quantum-policy-iteration}. 
    
    \indent We define the \textit{quantum policy evaluation} step as a quantum procedure (one can think of this as a unitary operation or a quantum circuit) that takes as input a classical policy $\pi$ and performs a mapping to create a quantum state that approximates or more generally contains some information about the classical value function $Q^\pi$. Again here one can define different quantum outputs that contain information about the value function and we will provide examples in the remaining of the paper. Similarly, we define the \textit{quantum policy improvement} step, as a quantum procedure (one may want to think of this as a generalized measurement operation), that takes as input the output quantum states from the quantum policy evaluation procedure, and extracts classical information by performing measurements on them, in order to compute a new policy $\pi'$ based on some policy update rule.
    
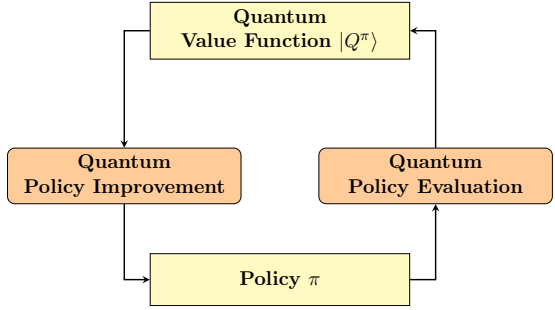
\begin{figure}[t]
    \centering
    \resizebox{0.45\linewidth}{!}{
            \begin{tikzpicture}[node distance=2cm,scale=0.5]
            \node[font=\bfseries] (function) [box_yellow] {Quantum \\ Value Function $\ket{Q^\pi}$};
            \node[font=\bfseries] (pi_improv) [box_orange, below of=function, ,yshift=-0.8cm, xshift=-3cm] {Quantum \\ Policy Improvement};
            \node[font=\bfseries] (pi_eval) [box_orange, right of=pi_improv, xshift=4cm] {Quantum \\ Policy Evaluation};
            \node[font=\bfseries] (policy) [box_yellow, below of=pi_improv, xshift=3cm] {Policy $\pi$};
            \draw [arrow] (pi_eval) |- (function);
            \draw [arrow] (function) -| (pi_improv);
            \draw [arrow] (pi_improv) |- (policy);
            \draw [arrow] (policy) -| (pi_eval);
            \end{tikzpicture}
                }
    \caption{Quantum Policy Iteration}
    \label{fig:scheme-quantum-policy-iteration}
\end{figure}

\subsection{Quantum policy evaluation}
\label{subsection:qpi-evaluation}

    \indent We will instantiate here our quantum policy iteration framework by providing one specific example of a quantum policy iteration method. We start by defining for any policy $\pi$ the quantum state: $$ \ket{Q^\pi}  = \frac{1}{\|Q^\pi\|}\sum_{sa} Q^\pi(s,a)\ket{s,a} $$ which encodes the classical value function $Q^\pi$ in the amplitudes of a normalized quantum state, and we call this state the quantum value function state. As we can see, this state contains information about the value function $Q^\pi$, though one needs to be careful since a copy of this state cannot recreate the complete function $Q^\pi$. It can still provide useful information, in particular we can use this state to sample a pair $(s,a)$ with probability proportional to $(Q^\pi(s,a))^2$. Moreover, there exist efficient quantum procedures for producing approximations to this state that we describe below. 
    
    \indent We assume we have quantum access to the MDP parameters $P$ and $R$ so that we can construct, for any policy $\pi$, a block-encoding of the policy transition matrix defined as: $$P^\pi:=\big[p(s,a,s') \pi(s',a')\big]_{sa,s'a'}$$ 
    
    \indent In Subsection \ref{subsection:qpi-block-encoding}, we discuss how one can get quantum access to the parameters of any MDP $\mathcal{M}$ and policy $\pi$ assuming that the transition matrix $P$ and the reward function $R$ can be efficiently computed, which is the case for many classes of environments. Denoting $\mathbf{A}^\pi = I-\gamma P^\pi$ and $\mathbf{b}=R$, it follows from the Bellman equation that the value function $Q^\pi$ is the solution of the linear system $\mathbf{A}^\pi Q^\pi=\mathbf{b}$. Using the \emph{quantum linear system solver} from Theorem \ref{theorem:quantum-linear-algebra}, we can build an $\epsilon$-approximation (in $\ell_2$-norm) state $\ket{\widehat{Q}^\pi}$ to the quantum value function state $\ket{Q^\pi}$: 

    \begin{thm}[Quantum policy evaluation]
    \label{theorem:quantum-policy-evaluation}
        Let $\mathcal{M}=(\mathcal{S},\mathcal{A},P,R,\gamma)$ be a finite Markov decision process, $\pi$ a  policy and $\epsilon > 0$ the precision parameter. Suppose there exists a $\mu_{P^\pi}$-block-encoding of the policy transition matrix $P^\pi$ that can be implemented with cost $\cost_{P^\pi}$. Also suppose we can prepare the reward vector $\ket{R}$ with cost $\cost_R$. Then there exists a quantum algorithm that with probability at least $(1-1/\poly(SA))$ returns a quantum state $\ket{\widehat{Q}^\pi}$  such that  $\|\ket{\widehat{Q}^\pi}-\ket{{Q}^\pi}\|\leq\epsilon$ with cost: $$ \bigO\paren*{\paren*{\mu_{P^\pi}\cost_{P^\pi}+\cost_R}\Gamma\polylog\paren*{\Gamma/\epsilon}} $$ 
    \end{thm}
    
\newpage

    \begin{proof}
        Using Theorem \ref{theorem:quantum-arithmetics}, we can implement with cost $\cost_{\mathbf{A}}=\bigO(\cost_{P^\pi})$ a $(1+\gamma \mu_{P^\pi})$-block-encoding of $\mathbf{A}^\pi=I-\gamma P^\pi$ as a linear combination of the trivial $1$-block-encoding of $I$ and the $\mu_{P^\pi}$-block-encoding of $P^\pi$. Then, we apply the quantum linear system solver from Theorem \ref{theorem:quantum-linear-algebra} with the procedure that generates the state $\ket{\mathbf{b}}=\ket{R}$ to return a quantum state $\ket{\widehat{Q}^\pi}$ $\epsilon$-close to $\ket{Q^\pi}=\ket{(\mathbf{A}^\pi)^{-1}\mathbf{b}}$ with cost: $$ \bigO\paren*{\kappa(\mathbf{A}^\pi)\paren*{\|\mathbf{A}^\pi\|^{-1}\paren*{1+\gamma\mu_{P^\pi}}\cost_{P^\pi}+\cost_R}\polylog\paren*{\kappa(\mathbf{A}^\pi)/\epsilon}} $$
        
        Since $P^\pi$ is a row-stochastic matrix, we have $ \|P^\pi \|\ = 1$ which implies that the singular values $\sigma_{\mathbf{A}}$ of $\mathbf{A}^\pi$ range in $[1-\gamma,1+\gamma]$.  We finish the proof by plugging the following upper bounds in the total cost: $1+\gamma \mu_{P^\pi} = \bigO(\mu_{P^\pi})$, $\kappa(\mathbf{A}^\pi) = \frac{\max \sigma_{\mathbf{A}}}{\min \sigma_{\mathbf{A}}}\leq \frac{1+\gamma}{1-\gamma}= \bigO(\Gamma)$ and $\|\mathbf{A}^\pi\|^{-1}\kappa(\mathbf{A}^\pi) =\frac{1}{\min \sigma_{\mathbf{A}}}\leq \frac{1}{1-\gamma}=\Gamma$.
    \end{proof}
    
    \indent We can now go ahead and define our specific quantum policy evaluation method as the one that given a policy $\pi$, uses a quantum linear system solver to create a quantum state which is an approximation of the state $\ket{Q^\pi}$. 

\subsection{Quantum policy improvement}
\label{subsection:qpi-improvement}

    \indent We will now describe a quantum policy improvement method that works together with the specific quantum policy evaluation method we described above, where approximations $\ket{\widehat{Q}^\pi}$ to $\ket{{Q}^\pi}$ are produced. Let us assume that these states can be produced in time $\cost_{Q^\pi}$. 
    
    \indent The quantum policy improvement method consists of first performing a number of $M$ measurements of the state $\ket{\widehat{Q}^\pi}$ for a total cost of $\bigO(M \times \cost_{Q^\pi})$. Denote by $M(s,a)$ the number of times outcome $(s,a)$ is observed. Then, we define the following strategy for the policy update: $$ \pi'(s) = \argmax_a M(s,a) \approx \argmax_a Q^\pi(s,a)$$ which takes time $\bigO(SA)$.
    
    \indent Let us analyze this policy improvement method. A measurement of the state $\ket{\widehat{Q}^\pi}$ outputs some state-action pair $(s,a)$ with probability $|\braket{\widehat{Q}^\pi}{s,a}|^2$ and we have: 
    
    \begin{equation*}
        \begin{aligned}
            \lim_{M\xrightarrow{}+\infty} & \frac{M(s,a)}{M} & = \quad& |\braket{\widehat{Q}^\pi}{s,a}|^2 & 
        \end{aligned}
    \end{equation*}
    
    \indent We make now some remarks on the appropriate value of $M$. Note first that the approximation state $\ket{\widehat{Q}^\pi}$ is $\epsilon$-close in $\ell_2$-norm to the quantum value function state $\ket{{Q}^\pi}$ and the cost of creating these states depends only logarithmically on the parameter $\epsilon$. Thus, we can take this parameter very small, and so if the number of measurements $M$ guarantees that we can closely reconstruct the state $\ket{\widehat{Q}^\pi}$, then this guarantee will carry over to the state $\ket{{Q}^\pi}$. Thus, if we want to guarantee an $\epsilon$-approximation of the quantum value function state $\ket{{Q}^\pi}$ in $\ell_2$-norm, one would need $M = \widetilde{\bigO}(SA/\epsilon^2)$, and for $\ell_\infty$-norm, which is what is used in reinforcement learning, $M = \widetilde{\bigO}(1/\epsilon^2)$. In this case, we are able to reconstruct an $\epsilon$-approximation in $\ell_\infty$-norm of the normalized quantum value function. 
    
    \begin{algorithm}[t]
    \caption{Quantum Policy Iteration}
    \label{algorithm:quantum-policy-iteration}
    \begin{algorithmic}
    \STATE {\bfseries input} MDP $\mathcal{M}$, number of measurements $M$, number of iterations $T$, precision $\epsilon$.
    \STATE initialize policy $\pi_0$.
    \FOR{$t=0$ {\bfseries to} $T-1$}
    \STATE initialize measurement histogram $M_t(s,a)=0$ for every pair $(s,a)$.
    \FOR{$m=0$ {\bfseries to} $M-1$}
        \STATE use the quantum linear solver with precision $\epsilon$ to obtain $\ket{\widehat{Q}^{\pi_t}} \approx \ket{(I-\gamma P^{\pi_t})^{-1}R}$.
        \STATE measure $\ket{\widehat{Q}^{\pi_t}}$ to get pair $(s,a)$ with probability $|\braket{\widehat{Q}^{\pi_t}}{s,a}|^2$.
        \STATE update measurement histogram $M_t(s,a)={}M_t(s,a)+1$.
    \ENDFOR
    \STATE improve policy as $\pi_{t+1}(s)=\argmax_a M_t(s,a)$ for every $s$.
    \ENDFOR
    \STATE {\bfseries output} policy $\pi_T$
    \end{algorithmic}
\end{algorithm}
    
    \indent In practice, setting $M$ to be $\widetilde{\bigO}(1/\epsilon^2)$ may be more than what is needed, since our goal is not to recreate $Q^\pi$ but to find $\pi'(s) = \argmax_a Q^\pi(s,a)$. This number $M$ can be adjusted in practice until the given method provides good results and in fact, in our experiments it was tuned to be significantly smaller than the theoretical value $\widetilde{\bigO}(1/\epsilon^2)$. 
    
    \indent It would be interesting to understand theoretically the number of samples needed for a successful implementation of a quantum policy improvement scheme, though we believe that in the end this would be use case-specific. Last, note that for the policy update rule we do not estimate directly the value function $Q^\pi$ but its normalized version $q^\pi$, i.e. we do directly the measurement outcomes $M(s,a)$, since the norm does not change the $\argmax$ calculation.
    
    \indent To sum up, we have defined a quantum policy improvement method as the one that given access to a quantum procedure that outputs a quantum value function state $\ket{Q^\pi}$, performs a number $M$ of measurements in order to create a measurement histogram from which the policy is updated via an $\argmax$ computation. 
    
    \indent The quantum policy iteration method we presented appears in Algorithm \ref{algorithm:quantum-policy-iteration}. In the next subsections we prove convergence, how to construct the necessary block-encodings for the  Markov decision process, and analyze its running time.

\subsection{Constructing block-encodings}
\label{subsection:qpi-block-encoding}
    
    \indent We are now going to show how to construct the block-encoding of $P^\pi$ and the unitary that prepares the quantum state associated to $R$ that we need in order to perform quantum policy iteration. In the classical case, one needs to have access to the transition matrix $P$, the policy $\pi$ and the reward vector $R$ in order to compute the policy-transition matrix $P^\pi$ and the corresponding value function $ Q^\pi=(I-\gamma P^\pi)^{-1} R = (\mathbf{A}^\pi)^{-1}\mathbf{b} $. In this subsection, we will discuss the access we need in the quantum case. More precisely, we will assume quantum access to the parameters $(\mathcal{S},\mathcal{A},P,R,\gamma)$ of the MDP $\mathcal{M}$ and to the policy $\pi$ which will be used to build the block-encoding of  $\mathbf{A}^\pi$ and the state $\ket{\mathbf{b}}$ as in Theorem \ref{theorem:quantum-policy-evaluation}. 
    
    \indent Next, we specify formally what we mean by quantum access to the MDP $\mathcal{M}$:
    
\newpage
    
    \begin{definition}[Quantum access to $\mathcal{M}$]
    \label{definition:block-mdp}
        Let $\mathcal{M}=(\mathcal{S},\mathcal{A},P,R,\gamma)$ be a finite MDP and let ${\rm c}_P\in\mathbb{R}_+$ such that $ {\rm c}_P \geq {\rm s}_{1}(P^\top)=\max_{s'}\sum_{sa} p(s,a,s') $. We say that we have quantum access to $\mathcal{M}$ with costs $(\cost_{P}, \cost_R)$ if we can implement: 
        \begin{enumerate}
            \item An oracle for the rows of the transition matrix with cost $\cost_P$ such that
            $$ \oracle_P^r: \ket{s,a}\ket{0_s} \longrightarrow \sum_{s'} \sqrt{p(s,a,s')} \ket{s,a}\ket{s'} $$
            
            \item An oracle for the columns of the transition matrix with cost $\cost_P$ such that:
            $$\oracle_P^c: \ket{0_s,0_a}\ket{s'} \longrightarrow \frac{1}{\sqrt{{\rm c}_P}} \Big[\sum_{sa} \sqrt{p(s,a,s')} \ket{s,a}\ket{s'} \Big] + \ket{G^\perp_{s'}}\ket{s'}$$
            
            where $\{\ket{G^\perp_{s'}}\}_{s'\in\mathcal{S}}$ are unnormalized garbage quantum states such that $\braket{G^\perp_{s'}}{s,a}=0$ for all transitions $(s,a,s') \in \mathcal{S}\times\mathcal{A}\times\mathcal{S}$.
            \item An oracle for the reward vector with cost $\cost_R$ such that:
            $$\oracle_R:\ket{0_s,0_a} \longrightarrow \ket{R} =  \frac{1}{\|R\|}\sum_{sa} r(s,a)\ket{s,a}$$
        \end{enumerate}
    \end{definition}
    
    \indent Let us look at the above oracles in more detail. The oracle $\oracle_R$ gives quantum access to the reward vector $\ket{R}$ whilst the oracles $\oracle_P^r$ and $\oracle_P^c$ correspond to the unitaries used to construct block-encodings from Lemma \ref{lemma:constructing-blocks} applied on the transition matrix $P$ where we have set the factor $p$ to be $1/2$ in our case. This choice is based on the observation that the rows of $P$ form probability distributions over $\mathcal{S}$ and can be seen as valid quantum states. The unitary $\oracle_P^r$ encodes the rows of $P$ and  maps every state-action pair $\ket{s,a}$ to the quantum state $\sum_{s'} \sqrt{p(s,a,s')} \ket{s'}$  with amplitudes forming a probability distribution, over next states $\ket{s'}$, that match the dynamics of the MDP $\mathcal{M}$. On the other hand, $\oracle_P^c$ encodes the columns of $P$ and maps every next-state $\ket{s'}$ to the neighboring state-action pairs $\ket{s,a}$ such that $p(s,a,s')\neq0$.  The quantity ${\rm c}_P$ depends on the dynamics of $\mathcal{M}$ and is an upper-bound on the portion of the state-action space $\mathcal{S}\times\mathcal{A}$ covered by every $s'$. 
    
    \indent The definitions of the above oracles are made so that Lemma \ref{lemma:constructing-blocks} with parameters $\alpha=1$ and $\beta=\sqrt{{\rm c}_P}$ implies that we can implement a $\sqrt{{\rm c}_P}$-block-encoding of $P$ using the oracles $\oracle_P^r$ and $\oracle_P^c$ and with cost $\cost_P$.
    
    \indent We have shown that efficient quantum access to $\mathcal{M}$ as defined in Definition \ref{definition:block-mdp} suffices to implement efficiently a block-encoding for the transition matrix $P$. In the next definition, we introduce quantum access to the policy $\pi$ where we map each state to a distribution over the set of actions. 
    
    \begin{definition}[Quantum access to $\pi$]
    \label{definition:block-policy}
        Let $\pi$ be a finite  policy. We say that we have quantum access to $\pi$ with cost $\cost_\pi$ if we can implement the following oracle with cost $\cost_\pi$:  
        $$\oracle_\pi:\ket{s,0_a} \longrightarrow  \ket{s, \pi(s)} = \sum_{a} \sqrt{\pi(s,a)} \ket{s,a} $$
    \end{definition}
    
    \indent Recall that we need a block-encoding for the matrix ${P^\pi}$ to perform quantum policy evaluation as in Theorem \ref{theorem:quantum-policy-evaluation}. We will build its block-encoding by combining the block-encoding of ${P}$ with the oracle $\oracle_\pi$ as stated in the following lemma:
    
    \begin{lem}[Block-encoding of $P^\pi$]
    \label{lemma:block-policy}
        Given quantum access to an MDP $\mathcal{M}$ with cost $(\cost_P,\cost_R)$ and to a policy $\pi$ with cost $\cost_\pi$, we can implement a $\mu_{P^\pi}$-block-encoding of ${P^\pi}$ with cost $\bigO(\cost_P + \cost_\pi)$, where the factor $\mu_{P^\pi}=\sqrt{{\rm c}_P}$ does not depend on the policy $\pi$.
    \end{lem}
    
    \begin{proof}
        We will use $\oracle_\pi$ to construct a $1$-block-encoding of the matrix $\Pi \in\mathbb{R}^{SA\times S} $ defined as: $$\Pi:= \big[\mathds{1}[s=s']\pi(s,a)\big]_{sa,s'}$$ and then use Lemma \ref{lemma:constructing-blocks} to get a block-encoding of $P^\pi$ which can be rewritten as $P^\pi=P\Pi^\top$. First, we build the oracle acting on the three registers $\ket{0_s}\ket{0_a}\ket{s'}$ that gives access to the columns of $\Pi$. Let $\unitary_C$ be the unitary that uses CNOT gates to copy the third register to the first register. If we apply $\unitary_C$ followed by $\oracle_\pi$ on the first register, we get: $$ \sum_a \sqrt{ \pi(s',a)} \ket{s',a} \ket{s'} = \sum_{sa} \sqrt{\mathds{1}[s=s'] \pi(s,a)} \ket{s,a} \ket{s'}$$
    
        Next, we build the oracles acting on four registers $\ket{s}\ket{a}\ket{0_s}\ket{0_a}$ that gives quantum access to the rows of $\Pi$. Let $\unitary_C$ be the unitary that copies the first state register to the third register. If we apply $\unitary_C$ followed by $\oracle_\pi$ on the last two registers, we get: 
        $$  \sum_{a'} \sqrt{\pi(s,a')}\ket{s,a}\ket{s}\ket{a'}$$
        Using again $\unitary_C$ to copy the second register to the fourth register we get
        $$
        \sqrt{\pi(s,a)}\ket{s,a}\ket{s}\ket{0} + \ket{G_{sa}^\perp}
        $$
        where $\ket{G_{sa}^\perp}$ is a garbage quantum state such that $\braket{G_{sa}^\perp}{s,a,s,0}=0$. The procedures above for accessing the rows and columns of $\Pi$ give us a $1$-block-encoding of $\Pi$. Finally, we apply the product of the block-encodings of  ${\Pi}$ and $P$ to get a $\mu_{P^\pi}$-block-encoding of $P\Pi^\top=P^\pi$. 
    \end{proof}
    
\subsection{Running time analysis}
\label{subsection:qpi-running-time}
    
    \indent From Theorem \ref{theorem:quantum-policy-evaluation}, we see that the cost of the quantum policy evaluation step is: $$ \bigO\paren*{\paren*{\mu_{P^\pi}(\cost_{P}+\cost_\pi)+\cost_R}\Gamma\polylog\paren*{\Gamma/\epsilon}} $$
    
    \indent Let us make some comments now of how this cost can behave in practice. We will see in following sections that for many MDP of interest quantum access can be implemented using quantum circuits of $\widetilde{\bigO}(SA)$ qubits and with only $\polylog(SA)$ depth, and thus the running time for the quantum policy evaluation, where here time refers to the depth of the quantum circuit will be of the form:
    $$ \bigO\paren*{\mu_{P^\pi}\Gamma\polylog\paren*{SA\Gamma/\epsilon}}$$
    
    \indent The overall running time of the quantum policy iteration will then be $\widetilde{\bigO}(SA+\mu_{P^\pi}M\Gamma)$, where $M$ is the number of measurements during the quantum policy improvement step. If we are using the $\ell_\infty$-tomography algorithm from Theorem \ref{theorem:tomography} then $M$ is taken to be $\widetilde{\bigO}(1/\epsilon^2)$ and the running time becomes $\widetilde{\bigO}(SA+\mu_{P^\pi}\Gamma/\epsilon^2)$, while in the experiments the value was actually smaller. 
    
    \indent In comparison, the running time of classical policy iteration is $\bigO((SA)^\omega)$ when using a classical linear system solver\footnote{$\omega$ is the matrix multiplication exponent, with best known theoretical value $2.37$ and in practice close to $3$.}. Whether our quantum algorithm provides an advantage for a specific environment depends on the environment parameters, i.e.\ the values of $\Gamma$ and $\mu_{P^\pi}$, what is the actual cost of constructing the block encodings of the transition matrix $P$ and the policy $\pi$, as well as how many samples $M$ are needed for a good policy improvement method. 
    
    \indent  Let us also remark on the value $\mu_{P^\pi}$. This value can be shown in the worst case to be $\sqrt{SA}$ but we expect it to be much smaller when we have efficient access to the transition matrix $P$ as detailed in Subsection \ref{subsection:qpi-block-encoding} where we have shown that $\mu_{P^\pi}=\sqrt{{\rm c}_P}$. The idea of the bound ${\rm c}_P$ is that in some of the most studied environments in reinforcement learning, any next state $s'$ arises as a result of taking an action $a$ from a small number of neighboring states $s$. In other terms, the transition matrix is very sparse and usually has $\bigO(A)$ non-zero elements in each column. We use our approach to build a $\sqrt{{\rm c}_P}$-block-encoding of $P^\pi$ for a total running of $\widetilde{\bigO}( \sqrt{{\rm c}_P} \times \Gamma/\epsilon^2)$. We expect this bound $\sqrt{{\rm c}_P}$ to be very small and have poly-logarithmic dependence on the size of the state space $\mathcal{S}$. For example, in the case of $d$-dimensional mazes, the factor ${\rm c}_P$ can be chosen such that ${\rm c}_P = 2d = \bigO(\log S)$. For two-player board games, such as chess or go, again the number of different states that could result to a particular state of the board through a single action are small and we can again think of it as ${\rm c}_P=\bigO(A)$, much smaller than the number of states that grows exponentially with the size of the board game. A concrete example is given in Subsection \ref{subsection:exp-frozen-lake} where we show, for \textsc{FrozenLake}, that ${\rm c}_P=\bigO(1)$ is constant and does not depend on the environment size. 
    
\subsection{Convergence guarantees}
\label{subsection:qpi-convergence}
    
    \indent We are going now to prove the theoretical convergence of our algorithm with precision $\epsilon$ and $M=\widetilde{\bigO}(1/\epsilon^2)$ measurements using a similar approach to \cite{Lagoudakis2003LeastSquaresPI}. Our bound is similar to the result given in Theorem \ref{theorem:bound-approximate-policy-iteration} but using the norm $\|.\|_{\rho}\leq\|.\|_\infty$ which is the $\ell_2$-norm weighted by the uniform distribution $\rho$ over $\mathcal{S}\times\mathcal{A}$. The $\|.\|_\rho$ norm is equal to the expected norm of a coordinate, instead of the maximum one as in the $\ell_\infty$ norm. Weighted quadratic norms are also used in classical reinforcement learning as in \cite{Munos2003ErrorBF} to prove the convergence of approximation algorithms. Next, we state the error bound on the policies generated by our algorithm when performing $M=\widetilde{\bigO}(1/\epsilon^2)$ measurements followed by a greedy update on the reconstructed normalized value function $\widehat{q}^{\pi_t}$ as described in Algorithm \ref{algorithm:quantum-policy-iteration}. 
    
    \begin{thm}[Error bound of QPI]
        \label{theorem:qpi-convergence}
        Let $\{\pi_t\}_{t\in\mathbb{N}}$ be the sequence of policies generated by the quantum policy iteration algorithm with a greedy update and let $\{\widehat{q}^{\pi_t}\}_{t\in\mathbb{N}}$ be the corresponding approximate normalized value functions. Then, this sequence satisfies the following suboptimality bound: $$ \limsup\limits_{t\rightarrow+\infty} \|\ket{Q^*}-\ket{Q^{\pi_t}}\|_{\rho} \leq 2\sqrt{2} \gamma\Gamma^2\limsup\limits_{t\xrightarrow{}+\infty}\|\widehat{q}^{\pi_t}-q^{\pi_t}\|_\infty$$
    \end{thm}
\newpage
    \begin{proof}
        First, we use Theorem \ref{theorem:bound-approximate-policy-iteration} on our quantum policy iteration algorithm which can be seen as a classical approximate policy iteration algorithm with a greedy update applied on $\widetilde{Q}^{\pi_t}=\|Q^{\pi_t}\|.\widehat{q}^{\pi_t}$ where $\widehat{q}^{\pi_t}(s,a):=\sqrt{M(s,a)/M}$. Moreover, note that the update in Algorithm \ref{algorithm:quantum-policy-iteration} is equivalent to $\pi_{t+1}(s)=\argmax\widehat{q}^{\pi_t}(s,a)$. In this case, the approximation errors $\|\widetilde{Q}^{\pi_t} - Q^{\pi_t}\|_\infty$ are bounded by $\|Q^*\|.\|\widehat{q}^{\pi_t}-q^{\pi_t}\|_\infty$ since $\|Q^{\pi_t}\|\leq\|Q^*\|$ for every policy $\pi_t$. Using Theorem \ref{theorem:bound-approximate-policy-iteration}, we have: $$ \limsup\limits_{t\xrightarrow{}+\infty} \|Q^*-Q^{\pi_t}\|_{\infty} \leq  2\gamma\Gamma^2 \limsup\limits_{t\xrightarrow{}+\infty} \|\widetilde{Q}^{\pi_t}-Q^{\pi_t}\|_{\infty} \leq 2 \gamma \Gamma^2 \|Q^*\| \limsup\limits_{t\xrightarrow{}+\infty}\|\widehat{q}^\pi-q^\pi\|_\infty $$

        Since we define our rewards to be greater than $0$, the angle between the vectors $Q^{\pi_t}$ and $Q^*$ will be no greater than $\pi/2$. Using Claim \ref{claim:angle-error}, we have: $$\|\ket{Q^*}-\ket{Q^{\pi_t}}\| \leq \frac{\sqrt{2}}{\|Q^*\|}\|Q^*-Q^{\pi_t}\|$$
        By taking the limit superior in the inequality above and observing that $\|.\|_{\rho}=\|.\|/\sqrt{SA}$, we conclude the proof by: 
        \begin{equation*}
            \begin{aligned}
                \limsup\limits_{t\rightarrow+\infty} \|\ket{Q^*}-\ket{Q^{\pi_t}}\|_{\rho} & \leq  \frac{\sqrt{2}}{\|Q^*\|}\limsup\limits_{t\rightarrow+\infty}\|Q^*-Q^{\pi_t}\|_{\rho} \\
                & \leq  \frac{\sqrt{2}}{\|Q^*\|}\limsup\limits_{t\rightarrow+\infty}\|Q^*-Q^{\pi_t}\|_{\infty} \\
                & \leq  2\sqrt{2}\gamma\Gamma^2\limsup\limits_{t\xrightarrow{}+\infty}\|\widehat{q}^{\pi_t}-q^{\pi_t}\|_\infty 
            \end{aligned}
        \end{equation*} 
    \end{proof}
    
    \indent Using weighted quadratic norms instead of the $\ell_\infty$-norm appears in many approximate algorithms in reinforcement learning when the subroutines minimize the $\ell_2$-norm \cite{Munos2003ErrorBF}. In our case, the bound shows that our quantum approach is a stable algorithm. When the inequality $\|\widehat{q}^{\pi_t}-{q}^{\pi_t}\|_\infty\leq\epsilon$ 
    holds for every iteration, the bound in Theorem \ref{theorem:qpi-convergence} becomes: 
    $$ \limsup\limits_{t\rightarrow+\infty} \|\ket{Q^*}-\ket{Q^{\pi_t}}\|_{\rho} \leq 2\sqrt{2}\gamma\Gamma^2 \epsilon$$
    
    which shows that quantum policy iteration oscillates between sub-optimal policies with value functions $\epsilon$-close to the optimal policy. 
    
\section{Quantum approximate policy iteration}
\label{section:approximate-quantum-policy-iteration}
    
\subsection{General framework for quantum approximate policy iteration}
\label{subsection:aqpi-framework}
    
    \indent We continue the description of our general framework for quantum reinforcement learning by looking at the common case where we may not be able to compute directly the value function $Q^\pi$, for example when the dynamics of the environment are unknown or the state-action space is too large. Instead, one approximates it with linear or non-linear value functions. Linear value functions correspond to the use of a linear combinations of features whilst the non-linear case corresponds to the use of non-linear approximation schemes such as neural networks for example. 
    
    \indent In our case, we are going to provide quantum algorithms for approximate policy iteration with linear value function approximation. We will start with the model-based case when having access to the parameters of the MDP $\mathcal{M}$ and provide a model-free implementation in Subsection \ref{subsection:aqpi-model-free} for the case when the dynamics of $\mathcal{M}$ are unknown. We refer to our framework as quantum approximate policy iteration (QAPI) and we summarize the general procedure in Figure \ref{fig:scheme-quantum-approximate-policy-iteration}.
    
\subsection{Quantum approximate policy evaluation}
\label{subsection:aqpi-evaluation}
    
    \indent In many cases, linear architectures are used for value function approximation where the $Q^\pi$ values are approximated by a linear combination $\widetilde{Q}^\pi$ of $K$ basis functions of the form $\phi_k: \mathcal{S}\times \mathcal{A} \xrightarrow[]{} \mathbb{R}$ and the policy parameters $w^\pi \in \mathbb{R}^K$:
    $$ Q^\pi(s,a) \approx \widetilde{Q}^\pi(s,a) = \sum\limits_{k} \phi_k(s,a) w^\pi_k $$
    
    \indent Note that the linearly independent features $\Phi:=[\phi_k(s,a)]_{sa,k} \in \mathbb{R}^{SA\times K}$ are  usually hand-crafted and common between all policies. On the other hand, one should compute $w^\pi \in \mathbb{R}^K$ to get the estimated value function $\widetilde{Q}^\pi= \Phi w^\pi$. To do so, we will base our approach on the least-squares policy iteration algorithm by Lagoudakis and Parr \cite{Lagoudakis2003LeastSquaresPI} for finding the parameters $w^\pi$. In the model-based case when we have access to $P$ and $R$, we can compute $P^\pi$ and we retrieve $w^\pi$ as the solution of the linear system $\mathbf{A}^\pi w^\pi = \mathbf{b}$ with  $\mathbf{A}^\pi= \Phi^\top ( \Phi- \gamma P^\pi \Phi)$ and $\mathbf{b}=\Phi^\top R$.
    
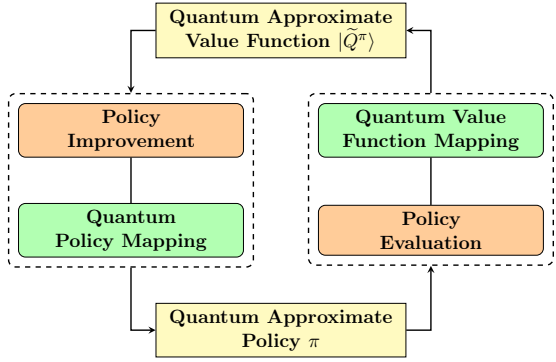
\begin{figure}[t]
    \centering
    \resizebox{0.45\linewidth}{!}{
        \begin{tikzpicture}[node distance=2cm]
            \node[font=\bfseries] (function) [box_yellow] {Quantum Approximate \\ Value Function $\ket{\widetilde{Q}^\pi}$};
            \node[font=\bfseries] (pi_improv) [box_orange, below of=function, xshift=-3cm] {Policy \\ Improvement};
            \node[font=\bfseries] (pi_actor) [box_green, below of=pi_improv] {Quantum \\ Policy Mapping};
            \node[font=\bfseries] (pi_eval) [box_orange, right of=pi_actor, xshift=4cm] {Policy \\ Evaluation};
            \node[font=\bfseries] (pi_critic) [box_green, right of=pi_improv, xshift=4cm] {Quantum Value Function Mapping};
            \node[font=\bfseries] (policy) [box_yellow, below of=pi_actor, xshift=3cm] {Quantum Approximate \\ Policy $\pi$};
            \node[font=\bfseries][draw, thick, rounded corners, dashed, inner xsep=0.5em, inner ysep=0.5em, fit= (pi_eval) (pi_critic)] (box_eval) {};
            \node[font=\bfseries][draw, thick, rounded corners, dashed, inner xsep=0.5em, inner ysep=0.5em, fit= (pi_improv) (pi_actor)] (box_improv) {};
            \draw [arrow] (box_eval) |- (function);
            \draw [arrow] (function) -| (box_improv);
            \draw [arrow] (box_improv) |- (policy);
            \draw [arrow] (policy) -| (box_eval);
            \draw[thick] (pi_actor) -- (pi_improv);
            \draw[thick] (pi_critic) -- (pi_eval);
        \end{tikzpicture}
                }
    \caption{Quantum Approximate Policy Iteration}
    \label{fig:scheme-quantum-approximate-policy-iteration}
\end{figure}
    
    \indent We now define for any policy $\pi$ the corresponding quantum state: $$ \ket{w^\pi} = \frac{1}{\|w^\pi\|}\sum_k w_k^\pi \ket{k}$$
    which encodes the classical weight vector $w^\pi$ in the amplitudes of a normalized quantum state. 
    
    \indent Using the same arguments provided in Section \ref{section:approximate-quantum-policy-iteration}, a copy of this quantum state cannot recreate $w^\pi$. However, it can still be used to provide useful information for quantum policy improvement as we will show in the next subsection. Using quantum linear algebra techniques, there exist efficient quantum procedures for producing approximations to $\ket{w^\pi}$ that we describe below.
    
\newpage

    \indent First, we assume we have quantum access to the model parameters $P$ and $R$ as in the quantum policy iteration algorithm. We also assume that we have quantum access to the features matrix $\Phi$ and we discuss how to do so in Subsection \ref{subsection:aqpi-block-encoding}. Denoting $\mathbf{A}^\pi= \Phi^\top ( \Phi- \gamma P^\pi \Phi)$ and $\mathbf{b}=\Phi^\top R$, the weight vector $w^\pi$ is the solution to the linear system $\mathbf{A}^\pi w^\pi = \mathbf{b}$. Using the quantum linear system solver from Theorem \ref{theorem:quantum-linear-algebra}, we can build an $\epsilon$-approximation (in $\ell_2$-norm) state $\ket{\widehat{w}^\pi}$ to the quantum weight vector $\ket{w^\pi}$:
    
    \begin{thm}[Model-based policy evaluation]
    \label{theorem:model-based-evaluation}
        Let $\mathcal{M}=(\mathcal{S},\mathcal{A},P,R,\gamma)$ a finite Markov decision process, $\Phi$ a features matrix with $\|\Phi\|=1$, $\pi$ a  policy and $\epsilon > 0$ the precision parameter. Suppose there exists a $\mu_{P^\pi}$-block-encoding of the policy transition matrix ${P^\pi}$ with cost $\cost_{P^\pi}$ and a $\mu_\Phi$-block-encoding of the features matrix ${\Phi}$ with cost $\cost_\Phi$. Also suppose that we can prepare the reward vector $\ket{R}$ with cost $\cost_R$. Then there exists a quantum algorithm that returns a quantum state $\ket{\widehat{w}^\pi}$ such that  $\|\ket{\widehat{w}^\pi}-\ket{{w}^\pi}\|\leq\epsilon$ with cost:  $$ \bigO\paren*{ \kappa_\Phi^2 \paren*{\paren*{\mu_\Phi^2\mu_{P^\pi}+\mu_\Phi\kappa_\Phi}\cost_\Phi+\mu_\Phi^2\mu_{P^\pi}\cost_{P^\pi}+\kappa_\Phi \cost_R} \Gamma \polylog\paren*{\kappa_\Phi\Gamma /\epsilon}} $$
    \end{thm}
    \begin{proof}
        As we said, denoting $\mathbf{A}^\pi= \Phi^\top ( \Phi- \gamma P^\pi \Phi)$ and $\mathbf{b}=\Phi^\top R$, the weight vector $w^\pi$ is the solution to the linear system $\mathbf{A}^\pi w^\pi = \mathbf{b}$.
        First, we implement the procedure that prepares $\ket{\mathbf{b}}=\ket{\Phi^\top R}$.  Using the $\mu_\Phi$-block-encoding of ${\Phi}$,  and picking $\epsilon' = (1+\gamma)\epsilon/(2\sqrt{2}\kappa_\Phi^2\Gamma)$, we can efficiently generate a state $\ket{\widehat{\mathbf{b}}}$ such that $\|\ket{\widehat{\mathbf{b}}}-\ket{\mathbf{b}}\|\leq \epsilon'$ with cost: 
        $$ 
        \cost_{\mathbf{b}}=\bigO\paren*{\kappa_\Phi \paren*{\mu_\Phi \cost_\Phi+\cost_R}\polylog\paren*{\kappa_\Phi\Gamma/\epsilon}}
        $$
        Next, we build the block encoding of ${\mathbf{A}^\pi}={\Phi^\top ( I- \gamma P^\pi )\Phi}$ using the block-encodings of ${P^\pi}$, ${\Phi}$ and ${\Phi^\top}$. Using the same approach as in Theorem \ref{theorem:quantum-policy-evaluation}, we can build a $(1+\gamma \mu_{P^\pi})$-block-encoding of ${I-\gamma P^\pi}$ with cost $\bigO(\cost_{P^\pi})$. Using Theorem \ref{theorem:quantum-arithmetics} to compute the block-encoding of the product $\Phi^\top(I-\gamma P^\pi)\Phi$, we get a $\mu_{\mathbf{A}}$-block-encoding of ${\mathbf{A}^\pi}$ with cost $\cost_{\mathbf{A}}=\bigO\paren{\cost_\Phi + \cost_{P^\pi}}$ such that $\mu_{\mathbf{A}}=\mu_\Phi^2(1+\gamma \mu_{P^\pi})=\bigO(\mu_\Phi^2\mu_{P^\pi})$. 
        
        Then, we apply the quantum linear solver with precision $\epsilon/2$ to generate a state $\ket{\widehat{w}^\pi}$ such that $\|\ket{\widehat{w}^\pi}-\ket{(\mathbf{A}^\pi)^{-1}\widehat{\mathbf{b}}}\|\leq\epsilon/2 $ with cost: 
        $$ \bigO\paren*{\kappa(\mathbf{A}^\pi)\paren*{\|\mathbf{A}^\pi\|^{-1}\mu_{\mathbf{A}}\cost_{\mathbf{A}}+\cost_{\mathbf{b}}}\polylog\paren*{\kappa(\mathbf{A}^\pi)/\epsilon}}$$

        Since $P^\pi$ is a row-stochastic matrix, we know that $\|P^\pi\|=1$ and that the singular values of $I-\gamma P^\pi$ range in $[1-\gamma,1+\gamma]$. Similarly, we have $\|\Phi\|=1$ and the singular-values of $\Phi$ range in $[1/\kappa_\Phi,1]$. We conclude that $\mathbf{A}^\pi$ has its singular values in $[(1-\gamma)/\kappa_\Phi^2, 1+\gamma]$. It follows that $\kappa(\mathbf{A}^\pi) = \bigO(\kappa_\Phi^2\Gamma)$ and $\|\mathbf{A}^\pi\|^{-1}\kappa(\mathbf{A}^\pi) = \bigO(\kappa_\Phi^2\Gamma)$. We plug these bounds into the total cost to get the total cost of our algorithm:
        $$
        \bigO\paren*{\kappa_\Phi^2\paren*{\mu_\Phi^2\mu_{P^\pi}\paren*{\cost_\Phi+\cost_{P^\pi}}+\kappa_\Phi \paren*{\mu_\Phi \cost_\Phi+\cost_R}}\Gamma\polylog\paren*{\kappa_\Phi\Gamma/\epsilon}}
        $$
        We conclude the proof by showing that $\ket{\widehat{w}^\pi}$ is $\epsilon$-close to $\ket{w^\pi}=\ket{(\mathbf{A}^\pi)^{-1}\mathbf{b}}$. We use Claim \ref{claim:angle-error} to show that, for a small value of $\epsilon$, $\|\ket{(\mathbf{A}^\pi)^{-1}\widehat{\mathbf{b}}}-\ket{(\mathbf{A}^\pi)^{-1}\mathbf{b}}\|\leq\sqrt{2}\kappa(\mathbf{A}^\pi)\epsilon'$. Then, we have:
        \begin{equation*}
            \begin{aligned}
                \|\ket{\widehat{w}^\pi}-\ket{w^\pi}\| & \leq \|\ket{\widehat{w}^\pi}-\ket{(\mathbf{A}^\pi)^{-1}\widehat{\mathbf{b}}}\|+\|\ket{(\mathbf{A}^\pi)^{-1}\widehat{\mathbf{b}}}-\ket{(\mathbf{A}^\pi)^{-1}\mathbf{b}}\| \\
                & \leq \epsilon/2 + \epsilon'\kappa_\Phi^2\Gamma = \epsilon
            \end{aligned}
        \end{equation*}
    \end{proof}

\begin{algorithm}[t]
    \caption{Quantum Approximate Policy Iteration}
    \label{algorithm:quantum-approximate-policy-iteration}
    \begin{algorithmic}
    \STATE {\bfseries input} MDP $\mathcal{M}$,  features $\Phi$, number of measurements $M$, number of iterations $T$, precision $\epsilon$.
    \STATE initialize policy $\pi_0$.
    \FOR{$t=1$ {\bfseries to} $T$}
        \FOR{$s \in \mathcal{S}$}
            \STATE initialize measurement histogram.
            $M_t(a)=0$ for every action $a$.
            \FOR{$m=0$ {\bfseries to} $M-1$}
                \STATE use quantum linear solver with precision $\epsilon$ to obtain $\ket{\widehat{w}^{\pi_t}}\approx\ket{( \Phi^\intercal\Phi- \gamma \Phi^\intercal P^{\pi_t} \Phi)^{-1}\Phi^\intercal R}$.
                \STATE use quantum linear algebra with precision $\epsilon$ to obtain $\ket{\widehat{Q}^{\pi_t}(s,\cdot)}\approx \ket{\Phi(s)\widehat{w}^{\pi_t}}$.
                \STATE measure $\ket{\widehat{Q}^{\pi_t}(s,\cdot)}$ to get action $a$ with probability $|\braket{a}{\widehat{Q}^{\pi_t}(s,.)}|^2$.
                \STATE update measurement histogram $M_t(a)=M_t(a)+1$
            \ENDFOR
            \STATE improve policy as $\pi_{t+1}(s)=\argmax M_t(a)$.
        \ENDFOR
    \ENDFOR
    \STATE {\bfseries output} policy $\pi_T$
    \end{algorithmic}
\end{algorithm}

\subsection{Quantum approximate policy improvement}
\label{subsection:aqpi-improvement}
    
    \indent We will now describe several quantum policy improvements methods that work together with the approximate quantum policy evaluation method we described above, where for each policy $\pi$ we estimate a weight vector $\ket{\widehat{w}^\pi}$. Let us assume that these states can be produced in time $\cost_{w^\pi}$. Again, we can assume a very small $\epsilon$ in the approximation guarantee of the states $\ket{\widehat{w}^\pi}$ and the states $\ket{{w}^\pi}$ (since it appears only inside a logarithm in the running time) and thus the approximation to the state $\ket{\widehat{w}^\pi}$ we will achieve through measurements will provide the same guarantees for the state $\ket{{w}^\pi}$ as well.
    
    \indent Our goal is to be able to compute a greedy policy with respect to the approximate value function $\widehat{Q}^\pi =  \Phi \widehat{w}^\pi$. Since our quantum procedure produces the normalized state $\ket{\widehat{w}^\pi}$, we are going to perform measurements in order to compute the actions corresponding to the improved policy $\pi'$ defined as $\pi'(s)=\argmax_a \Phi(s) \widehat{w}^\pi$ where $\Phi(s)\in\mathbb{R}^{A\times K}$ is the matrix with rows $\Phi(s,a)^\top$ containing the features associated to the state $s$ such that $\Phi(s) \widehat{w}^\pi$ is an approximation to $Q^\pi(s,.)$. Next, we will describe three different improvement strategies. 
    
    \indent The first approach is similar to the one detailed in Subsection \ref{subsection:qpi-improvement} but requires an additional step. Since we have quantum access to $\Phi$, we use the quantum matrix multiplication procedure from Theorem \ref{theorem:quantum-linear-algebra} with the $\mu_\Phi$-block-encoding of $\Phi$ and the output $\ket{\widehat{w}^\pi}$ of the approximate quantum policy evaluation procedure to compute the quantum state $\ket{\Phi \widehat{w}^\pi}$ which is an approximate to the quantum value function $\ket{Q^\pi}$ with cost $\widetilde{\bigO}( \kappa_\Phi (\mu_\Phi \cost_\Phi + \cost_{w^\pi}))$. We then perform measurements on this quantum state and update the policy according to the rule: $$\pi'(s)=\argmax_a M(s,a) \approx \argmax_a \braket{\Phi w^\pi}{s,a}$$ where $M$ is the histogram of the measured state-action pairs $\ket{s,a}$ sampled from $\ket{\Phi \widehat{w}^\pi}$. The total cost of this policy update rule is $\widetilde{\bigO}(M \kappa_\Phi (\mu_\Phi \cost_\Phi + \cost_{w^\pi}))$ where the number of measurements can be adjusted in practice according to the arguments provided in Subsection \ref{subsection:qpi-improvement}.

    \indent The second approach reconstructs classically an approximation to the output $\ket{\widehat{w}^\pi}$ in order to improve the actual policy. First, we will perform a number of $M$ measurements on $\ket{\widehat{w}^\pi}$ such that we will sample for each measurement some feature index $\ket{k}$ with probability $|\braket{\widehat{w}^\pi}{k}|^2$. Since the components of $\ket{\widehat{w}^\pi}$ are not necessarily positive, we also need to perform sign estimation of the components of $\ket{w^\pi}$ by performing an additional number of $M$ measurements that query $\ket{\widehat{w}^\pi}$ \cite{Kerenidis2020QuantumAF}. Denoting by $M(k)$ the number of times the feature index $\ket{k}$ was sampled and by $\sigma(k)$ the estimated sign of $\braket{\widehat{w}^\pi}{k}$, the normalized vector with coordinates $\sigma(k)\sqrt{M(k)/M}$ is an approximation to the quantum state $\ket{w^\pi}$. Hence, we can use the following policy improvement rule:
    $$ \pi'(s) = \argmax_a \sum_{k} \sigma(k)\sqrt{M(k)} \phi_k(s,a) \approx \argmax_a \braket{\Phi(s,a)}{w^\pi} $$ The total cost of this policy improvement strategy is $\widetilde{\bigO}(SK + M \cost_{w^\pi})$ since we need to perform $M$ measurements in order to reconstruct classically $\ket{\widehat{w}^\pi}$ before performing $\bigO(K)$ operations to compute $\pi'(s)$ for each $s\in \mathcal{S}\text{ or }\mathcal{D}$.
    
    \indent The third approach consists of building approximations to the quantum states $\ket{\Phi(s)\widehat{w}^\pi}$ for every state $s$ using quantum matrix-vector multiplication and performing measurements on these quantum states. First, for every state $s$, we construct a $\mu_{\Phi(s)}$-block-encoding of $\Phi(s)$ that we apply to $\ket{\widehat{w}^\pi}$ in order to compute $\ket{\Phi(s)\widehat{w}^\pi}$ which is as an approximation to $\ket{Q^\pi(s,.)}$ defined as: $$ \ket{Q^\pi(s,.)} = \frac{1}{\|Q^\pi(s,.)\|} \sum_{a} Q^\pi(s,a) \ket{a}$$ Second, we measure the quantum states $\ket{\Phi(s)\widehat{w}^\pi}$ to get an action $a$ with probability $|\braket{a}{\widehat{Q}^\pi(s,.)}|^2 \approx Q^\pi(s,a)^2/\|Q^\pi(s,\cdot)\|^2$. Similarly to the approach in Subsection \ref{subsection:qpi-improvement}, we construct for every $s$ a histogram of measurements denoted by $M$ such that $M(a)$ is the number of times we measured action when applying the block-encoding of $\Phi(s)$ to $\ket{w^\pi}$. Then, we update the policy according to the rule: $$\pi'(s)=\argmax_a M(a) \approx \argmax_a \braket{a}{\Phi(s)w^\pi}$$ Let $\kappa_{\Phi|\mathcal{S}}$ and $\mu_{\Phi|\mathcal{S}}$ be upper bounds on the quantities $\kappa_{\Phi(s)}$ and $\mu_{\Phi(s)}$ of $\Phi(s)$ over all states $s$, the total cost for updating the policy is then $\bigO(MS\kappa_{\Phi|\mathcal{S}}(\mu_{\Phi|\mathcal{S}} \cost_{\Phi|\mathcal{S}} + \cost_{w^\pi}))$ since
    the cost for producing a single quantum state $\ket{\Phi(s)\widehat{w}^\pi}$ is $\bigO(\kappa_{\Phi(s)}(\mu_{\Phi(s)} \cost_{\Phi(s)} + \cost_{w^\pi}))$. 
    
    \indent We have defined different quantum approximate policy improvement methods that can be used together with the quantum policy evaluation described in previous sections. The quantum approximate policy iteration method with the third improvement strategy, which provides a good method for near term implementations, is used in Algorithm \ref{algorithm:quantum-approximate-policy-iteration}. In the next subsections, we are going to discuss how to construct the block-encodings of the features matrix $\Phi$, analyze the running time of our approach and provide a model-free implementation. 
    
\subsection{Constructing block-encodings}
\label{subsection:aqpi-block-encoding}
    
    \indent We are going to show how to build quantum access to the parameters required by Theorem \ref{theorem:model-based-evaluation}. We need to construct the block-encodings of the transition matrix $P^\pi$ and the features matrix $\Phi$ and build quantum access to the reward vector $\ket{R}$. 
    
    \indent We assume quantum access to $\mathcal{M}$ as in Definition \ref{definition:block-mdp} for the parameters of the MDP and we have already discussed in Subsection \ref{subsection:qpi-block-encoding} how to get a block-encoding for $P^\pi$ and the procedure that prepares  $\ket{R}$. In the following, we apply a similar approach to get a block-encoding for $\Phi$. We extend Definition \ref{definition:block-mdp} to the approximate case it by assuming access to an additional oracle that encodes the features: 
    
    \begin{definition}[Model-based quantum access]
    \label{definition:block-basis}
        Let $\Phi \in \mathbb{R}^{SA\times K}$ be a features matrix such that $\|\Phi(s,a)\|=1$ for every state-action pair $(s,a)$. We say that we have quantum access in the model-based case with cost $(\cost_P,\cost_R,\cost_\Phi)$ if, additionally to the oracles in Definitions \ref{definition:block-mdp} and \ref{definition:block-policy} with cost $(\cost_P,\cost_R)$, we can implement with cost $\cost_\Phi$ the following oracle and its controlled version for the features matrix $\Phi$:  
        $$\oracle_\Phi:\ket{s,a}\ket{0_k} \longrightarrow  \ket{s,a}\ket{\Phi(s,a)} = \sum_{sa} \phi_k(s,a) \ket{s,a}\ket{k} $$
    \end{definition}
    
    \indent Then, we use the oracle $\oracle_\Phi$ to build the block-encoding of $\Phi$ as shown in the following lemma:  
    
    \begin{lem}[Block-encoding of $\Phi$]
    \label{lemma:block-model-based} 
        Given quantum access to $\Phi$ with cost $\cost_\Phi$, we can implement a $\sqrt{K}$-block-encoding of $\Phi$ with cost $\bigO(\cost_\Phi)$.
    \end{lem}
    
    \begin{proof}
        Since the rows of $\Phi$ are normalized, we get from $\text{O}_\Phi$ a $\sqrt{K}$-block-encoding of $\Phi$ using Lemma \ref{lemma:constructing-blocks} with $p=1$.
    \end{proof}
    
\subsection{Running time analysis}
\label{subsection:aqpi-running-time}
    
    \indent We have formally defined the oracles that we need for the implementation of quantum approximate policy iteration and we are going to analyze its running time. From Theorem \ref{theorem:model-based-evaluation}, we see that the cost of quantum approximate policy evaluation in the model-based case is:
    $$\bigO\paren*{ \kappa_\Phi^2
        \paren*{\paren*{\mu_\Phi^2\mu_{P^\pi}+\mu_\Phi\kappa_\Phi}\cost_\Phi+\mu_\Phi^2\mu_{P^\pi}(\cost_{P}+\cost_{\pi})+\kappa_\Phi \cost_R} \Gamma\polylog\paren*{\kappa_\Phi\Gamma /\epsilon}}$$
    
    \indent As we said, we could make the assumption that oracles for the transition matrix $P$ and the policy $\pi$ can be built in poly-logarithmic depth, and the same for the feature matrix $\Phi$ that is hand-picked by us. We also have that the normalizing factor $\mu_\Phi=\sqrt{K}$.  We then have the following simplification of the running time of the model-based approximate quantum policy evaluation: $$ \cost_{w^\pi} = \bigO\paren*{\kappa_\Phi^2\paren*{K\mu_{P^\pi}+\sqrt{K}\kappa_\Phi}\Gamma \polylog\paren*{\kappa_\Phi SA\Gamma K/\epsilon}} $$
    
    \indent The overall running time of our algorithm where we apply the greedy update rule as in Algorithm \ref{algorithm:quantum-approximate-policy-iteration} will be $ \widetilde{\bigO} (\kappa_{\Phi|\mathcal{S}} MS\cost_{w^\pi})$. Classically, the running time of approximate policy iteration is $\bigO(SAK^2)$ for the approximate policy evaluation step and $\bigO(SAK)$ for the policy improvement step. Whether our algorithm provides an advantage over the classical one depends on the number of measurements $M$ required for policy improvement and the properties of the features function $\Phi$, namely the dimension $K$ and the condition numbers $\kappa_\Phi$ and $\kappa_{\Phi|\mathcal{S}}$, which given that we pick the matrix $\Phi$ ourselves, we can easily control. Moreover, we do not expect the number of measurements to grow with the size of the state space $S$ since we measure quantum states $\ket{Q^\pi(s,.)}$ of size $A$ which was not the case with $\ket{Q^\pi}$ of size $SA$. 
    
\subsection{Model-free implementation}
\label{subsection:aqpi-model-free}
    
    \indent We have defined a quantum algorithm for performing model-based approximate policy iteration where we have access to a model for the MDP $\mathcal{M}$. Next, we are going to show that we can also implement a model-free approach that does not require such access. When $P$ and $R$ are unknown, we assume having access to a source $\mathcal{D}$ containing transition samples of the form $(\tilde{s},\tilde{a},\tilde{s}',\tilde{r})$ and we compute an estimate $\widetilde{w}^\pi$ of $w^\pi$ as a solution to $\mathbf{A}^\pi \widetilde{w}^\pi = \mathbf{b}$ with  $\mathbf{A}^\pi= \widetilde{\Phi}^\top (\widetilde{\Phi}-\gamma \widetilde{P^\pi\Phi})$ and $\mathbf{b}=\widetilde{\Phi}^\top\widetilde{R}$ such that $\widetilde{\Phi} \in \mathbb{R}^{D\times K}$, $\widetilde{P^\pi\Phi}\in \mathbb{R}^{D\times K}$ and $\widetilde{R}\in \mathbb{R}^D$ are estimated using the samples from the source $\mathcal{D}$. Denoting by $D$ the number of samples in $\mathcal{D}$, $\widetilde{\Phi} \in \mathbb{R}^{D\times K}$  is the matrix with rows $\phi(\tilde{s}_i,\tilde{a}_i)^\top$ where $i$ denotes the $i$-th sample of $\mathcal{D}$, $\widetilde{P^\pi\Phi}\in \mathbb{R}^{D\times K}$ is the matrix with rows $\phi(\tilde{s}'_i,\pi(\tilde{s}'_i))^\top$ and $\widetilde{R}\in\mathbb{R}^{D}$ is the vector with elements $\widetilde{R}_i = \widetilde{r}_i$. In the quantum case, we will assume having quantum access to these three quantities and use the quantum linear algebra techniques to build an $\epsilon$-approximation (in $\ell_2$-norm) state $\ket{\widehat{w}^\pi}$ to $\ket{\widetilde{w}^\pi}$: 
    
    \begin{thm}[Model-free evaluation]
    \label{theorem:model-free-evaluation}
        Let $\mathcal{M}=(\mathcal{S},\mathcal{A},P,R,\gamma)$ be a finite or non-finite Markov decision process with unknown model $P$ and $R$, $\Phi$ a  features function such that $\|\Phi(s,a)\|=1$ for every state-action pair $(s,a)$, $\pi$ a deterministic  policy and $\epsilon > 0$ the precision parameter. Suppose there exists a $\mu_{\widetilde{\Phi}}$-block-encoding of the estimated $\widetilde{\Phi}$ and $\widetilde{P^\pi\Phi}$ with cost $T_{\widetilde{\Phi}}$. Both matrices having singular values ranging in $[1/\kappa_{\widetilde{\Phi}},1]$. Also suppose that we can prepare the estimated reward vector $\ket{\widetilde{R}}$ with cost $\cost_{\widetilde{R}}$. Then there exists a quantum algorithm that returns a quantum state $\ket{\widehat{w}^\pi}$ such that  $\|\ket{\widehat{w}^\pi}-\ket{\widetilde{w}^\pi}\|\leq\epsilon$ with cost:  
        $$ \bigO\paren*{\kappa_{\widetilde{\Phi}}^2\paren*{\mu_{\widetilde{\Phi}}^2\cost_{\widetilde{\Phi}}+\kappa_{\widetilde{\Phi}} \mu_{\widetilde{\Phi}} \cost_{\widetilde{\Phi}}+\kappa_{\widetilde{\Phi}}\cost_R}
        \Gamma\polylog\paren*{\kappa_{\widetilde{\Phi}}\Gamma /\epsilon}} $$
    \end{thm}
    
    \begin{proof}
        Using a similar approach to Theorem \ref{theorem:model-based-evaluation}, we can implement a $\mu_{\widetilde{\Phi}}$-block-encoding of $\widetilde{\Phi}^\top$ and generate an $\epsilon' = (1+\gamma)\epsilon/(2\sqrt{2}\kappa_{\widetilde{\Phi}}^2\Gamma)$ approximation state $\ket{\mathbf{\widehat{b}}}$ to the state $\ket{\mathbf{b}}=\ket{\widetilde{\Phi}^\top\widetilde{R}}$ with cost: $$\cost_{\mathbf{b}}=\bigO\paren*{\kappa_{\widetilde{\Phi}} \paren*{\mu_{\widetilde{\Phi}} \cost_{\widetilde{\Phi}}+\cost_{\widetilde{R}}}\polylog\paren*{\kappa_{\widetilde{\Phi}}\Gamma/\epsilon}}$$
        
        Next, we build the block encoding of ${\mathbf{A}^\pi}=\widetilde{\Phi}^\top (\widetilde{\Phi}-\gamma \widetilde{P^\pi\Phi})$ using the block-encodings of $\widetilde{\Phi}$, $\widetilde{\Phi}^\top$ and $\widetilde{P^\pi\Phi}$.
        First, note that we can construct $\mu^2_{\widetilde{\Phi}}$-block-encodings of $\widetilde{\Phi}^\top \widetilde{\Phi}$ and ${\paren{\widetilde{\Phi}^\top \widetilde{P^\pi\Phi}}}$ with cost $\bigO(\cost_{\widetilde{\Phi}})$. Using Theorem \ref{theorem:quantum-arithmetics}, we can implement a $\mu^2_{\widetilde{\Phi}}(1+\gamma)$-block-encoding of ${\paren{\widetilde{\Phi}^\top \widetilde{\Phi}-\gamma \widetilde{\Phi}^\top\widetilde{P^\pi\Phi}}}$ with cost $\cost_{\mathbf{A}}=\bigO(\cost_{\widetilde{\Phi}})$. Then, we apply the quantum linear solver with precision $\epsilon/2$ to generate a state $\ket{\widehat{w}^\pi}$ such that $\|\ket{\widehat{w}^\pi}-\ket{(\mathbf{A}^\pi)^{-1}\widehat{\mathbf{b}}}\|\leq\epsilon/2 $ with cost: 
        $$ \bigO\paren*{\kappa(\mathbf{A}^\pi)\paren*{\|\mathbf{A}^\pi\|^{-1}\mu_{\widetilde{\Phi}}^2\paren*{1+\gamma}\cost_{\widetilde{\Phi}}+\cost_{\mathbf{b}}}\polylog\paren*{\kappa(\mathbf{A}^\pi)/\epsilon}}$$
    \end{proof}
    
    \indent The model-free quantum policy evaluation approach above can work together with any of the improvement strategies described in Subsection \ref{subsection:aqpi-improvement}. The only difference is that we need to update the policy for all states $s' \in \mathcal{D}$, i.e. we iterate over all next-states $s'$ and update the estimated $\widetilde{P^\pi\Phi}$. The cost analysis is sill valid by replacing the space state size $S$ by the source size $D$. Next, we describe how to construct the necessary block-encodings in the model-free case. 
    
    \indent We are provided with a source $\mathcal{D}$ of transition samples  of the form $(\tilde{s},\tilde{a},\tilde{s}',\tilde{r})$ and their corresponding features. Similarly to the model-based case, we will assume that the features are normalized for each state-action pair $(\tilde{s},\tilde{a})$. The following definition gives the list of oracles that we need to implement in order to perform quantum approximate policy evaluation needed for Theorem \ref{theorem:model-free-evaluation}: 
    
    \begin{definition}[Model-free quantum access]
    \label{definition:block-model-free}
        Let $\mathcal{M}=(\mathcal{S},\mathcal{A},P,R,\gamma)$ be a finite or non-finite Markov decision process with unknown model $P$ and $R$, $\mathcal{D}$ a finite source of transition samples from $\mathcal{M}$ of the form $(\tilde{s},\tilde{a},\tilde{s}',\tilde{r})$, $\Phi$ a feature function such that $\|\Phi(\tilde{s},\tilde{a})\|=1$ for all $(\tilde{s},\tilde{a})$ and $\pi$ a deterministic  policy. We say that we have quantum access in the model-free case with costs $(\cost_{\widetilde{P}},\cost_{\widetilde{R}}, \cost_{\widetilde{\Phi}}, \cost_{\widetilde{\pi}})$ if we can implement: 
        \begin{enumerate}
            \item Two oracles for the transition samples with cost $\cost_{\widetilde{P}}$ such that:
            $$ \oracle_{\widetilde{P}}^{s,a}: \ket{i} \ket{0_s,0_a} \longrightarrow  \ket{i}\ket{\tilde{s}_i,\tilde{a}_i} $$
            $$ \oracle_{\widetilde{P}}^{s'}: \ket{i} \ket{0_s} \longrightarrow  \ket{i}\ket{\tilde{s}'_i} $$
            
            \item An oracle for the reward samples with cost $\cost_{\widetilde{R}}$ such that:
            $$ \oracle_{\widetilde{R}}:\ket{0_i} \longrightarrow \ket{\widetilde{R}} =  \frac{1}{\|\widetilde{R}\|}\sum_{i} \tilde{r}_i\ket{i} $$
            
            \item An oracle for the features function $\Phi$ with cost $\cost_{\Phi}$ such that:
            $$\oracle_{\Phi}: \ket{\tilde{s},\tilde{a}}\ket{0_k}\longrightarrow \ket{\tilde{s},\tilde{a}}\ket{\Phi(\tilde{s},\tilde{a})} = \sum_k \phi_k(\tilde{s},\tilde{a}) \ket{\tilde{s},\tilde{a}}\ket{k}$$
            
            \item An oracle for the deterministic policy $\pi$ with cost $\cost_{\widetilde{\pi}}$ such that: 
            $$\oracle_{\tilde{\pi}}: \ket{\tilde{s},0_a} \longrightarrow \ket{\tilde{s},\pi(\tilde{s})}$$
            
        \end{enumerate}
    \end{definition}
    
    \indent Our quantum policy evaluation algorithm requires a procedure for the estimated vector $\ket{\widetilde{R}}$, which is given by the oracle $\oracle_{\widetilde{R}}$, and the block-encodings of $\widetilde{\Phi}$ and  $\widetilde{P^\pi\Phi}$ given by the following lemma:
    
    \begin{lem}[Block-encodings of $\widetilde{\Phi}$ and  $\widetilde{P^\pi\Phi}$]
    \label{lemma:block-model-free} 
        Given quantum access in the model-free case to $\mathcal{D}$, $\Phi$ and $\pi$ as in Definition \ref{definition:block-model-free}, we can implement a $\sqrt{K}$-block-encoding of $\widetilde{\Phi}$ and $\widetilde{P^\pi\Phi}$ with cost $\cost_{\widetilde{\Phi}}=\bigO(\cost_{\widetilde{P}}+\cost_\Phi+\cost_{\widetilde{\pi}})$.
    \end{lem}
    
    \begin{proof}
    If we start from the state $\ket{i}\ket{0_k}\ket{0_s,0_a}$ and apply $\oracle_{\widetilde{P}}^{s,a}$ on the first and third registers followed by $\oracle_\Phi$ on the third and second register, we get the following mapping after uncomputing the third register using the adjoint operation $(\oracle_{\widetilde{P}}^{s,a})^\dag$: $$ \ket{i}\ket{0_k} \longrightarrow \sum_k \phi_k(\tilde{s}_i,\tilde{a}_i) \ket{i}\ket{k}$$
    
    Similarly, if we start from the state $\ket{i}\ket{0_k}\ket{0_s,0_a}$ and apply $\oracle_{\widetilde{P}}^{s'}$ followed by $\oracle_{\widetilde{\pi}}$ and  $\oracle_\Phi$, we get the following mapping after uncomputing the last register using $(\oracle_{\widetilde{P}}^{s'})^\dag$: $$ \ket{i}\ket{0_k} \longrightarrow \sum_k \phi_k(\tilde{s}'_i,\pi(\tilde{s}'_i)) \ket{i}\ket{k} $$
    Then, we use Lemma \ref{lemma:constructing-blocks} to construct respectively the block-encodings for $\widetilde{\Phi}$ and  $\widetilde{P^\pi\Phi}$ by setting the factor $p$ to be $1$.
    \end{proof}
    
    \indent Assuming that the circuits for model-free quantum access can be implemented in poly-logarithmic depth as in Subsection \ref{subsection:aqpi-running-time}, the running time of approximate quantum policy iteration simplifies to: 
    $$ \bigO\paren*{\kappa_{\widetilde{\Phi}}^3\sqrt{K}\Gamma \polylog\paren*{\kappa_{\widetilde{\Phi}} DAK\Gamma/\epsilon}} $$
    
\section{Applications}
\label{section:applications}
    
    \indent In the previous sections, we formulated our quantum policy iteration algorithms using the block-encoding framework and we have explicitly described what quantum oracles we need in order to construct these block-encodings. In this section, we will describe how to implement in practice quantum access to those oracles for the \textsc{FrozenLake} and \textsc{InvertedPendulum} which are two environments listed in OpenAI's Gym \cite{Brockman2016OpenAIG} and widely used in reinforcement learning. 
    
\subsection{Application to \textsc{FrozenLake}}
\label{subsection:exp-frozen-lake}
    
    \indent \textbf{Description of the environment:} \textsc{FrozenLake} is an environment that consists of a two-dimensional grid of size $X\times Y$ where the agent moves around the grid in four directions to reach the goal state without falling into holes. The episode terminates if the agent steps into a hole or reaches the goal state where a reward of $+1$ is perceived. Its state space $\mathcal{S}=\{(x,y)|x\in[X],y\in[Y]\}$ is the set of all grid positions and its action space $\mathcal{A}=\{(0,1),(0,-1),(1,0),(-1,0)\}$ contains the four possible actions: \textit{up}, \textit{down}, \textit{left} and \textit{right}. For example, taking action $a=(0,-1)$ when in state $s=(x,y)$ moves the agent to the next state $s'=s+a=(x,y-1)$.  
    
    \indent \textbf{Quantum access:} We want to build quantum access to the MDP $\mathcal{M}=(\mathcal{S},\mathcal{A},P,R,\gamma)$ associated to this environment by constructing the oracles as in Definition \ref{definition:block-mdp}. In the classical case, we can recover the environment dynamics $(P,R)$ by specifying the goal state $s_G$ and the subset of walkable positions $\mathcal{F}\subset\mathcal{S}$, i.e. positions that are neither holes nor the goal state. Similarly, we show in the following claim that we can build quantum access to $\mathcal{M}$ if we have access to appropriate oracles that encode the subset $\mathcal{F}$ and the state $s_G$: 
    
    \begin{claim}\label{claim-frozen-lake} 
         Let $\mathbbm{1}_\mathcal{F}$ be the indicator function for the subset of walkable positions in the grid. Given quantum access to an oracle $\oracle_\mathcal{F}: \ket{s}\ket{0}\xrightarrow{}\ket{s}\ket{\mathbbm{1}_\mathcal{F}(s)}$ with cost $\cost_\mathcal{F}$ and to an oracle $\oracle_G: \ket{0_s}\xrightarrow[]{}\ket{s_G}$ with cost $\cost_G$, we can build quantum access to $\mathcal{M}$ with costs  $(\cost_P,\cost_R)=(\bigO(\cost_\mathcal{F}),\bigO(\cost_G))$.
    \end{claim}
\newpage
    \begin{proof}
        We will build quantum access to $\mathcal{M}$ by constructing the oracles $\oracle_P^r$, $\oracle_P^c$ and $\oracle_R$ from Definition \ref{definition:block-mdp}. Classically, the transition matrix $P$ and the reward vector $R$ can both be recovered from $\mathcal{F}$ and $s_G$ since: 
        $$ p(s,a,s')= 
            \begin{cases}
                1, & \text{if}\ (s\in \mathcal{F} \ \text{and} \ s'=s+a) \ \text{or} \ (s\not\in \mathcal{F} \ \text{and} \ s'=s)\\
                0, & \text{otherwise}\\
            \end{cases}  $$
        $$ r(s,a)= 
            \begin{cases}
                1, & \text{if}\ (s\in \mathcal{F}\ \text{and} \ s+a =s_G) \\
                0, & \text{otherwise}\\
            \end{cases} $$
        First, we will construct the oracle $\oracle_P^r$. If we assume without loss of generality that all positions located in the borders of the grid are non-walkable, then the oracle $\oracle_P^r$ that encodes the rows of $P$ corresponds to the mapping: 
        $$ \oracle_P^r: \ket{s,a}\ket{0_s} \longrightarrow 
        \begin{cases}
                \ket{s,a}\ket{s+a}, & \text{if}\ s\in \mathcal{F}\\
                \ket{s,a}\ket{s}, & \text{if}\ s\not\in \mathcal{F}\\
        \end{cases} $$
        Both mappings $\ket{s,a}\ket{0_s}\xrightarrow[]{}\ket{s,a}\ket{s+a}$ and $\ket{s,a}\ket{0_s}\xrightarrow[]{}\ket{s,a}\ket{s}$ can be implemented in linear cost on the number of qubits used to represent state-action pairs $\ket{s,a}$. If we combine both mappings with $\oracle_\mathcal{F}$ applied on one ancilla qubit, we can construct the oracle $\oracle_P^r$ with cost $\bigO(\cost_\mathcal{F})$.
    
    \begin{figure}[t]
    \centering
    \resizebox{0.25\linewidth}{!}{
        \begin{tikzpicture}
            \draw[step=1cm, gray, very thin] (0, -8) grid (8, 0);
            \foreach \x in {0,1,...,7} {
                \node at (-0.5,-\x-0.5) {$\x$};
                \node[left] at (\x+0.5,0.5) {$\x$};
                \fill[black] (0, -\x) rectangle (1, -\x-1);
                \fill[black] (7, -\x) rectangle (8, -\x-1);
                }
            \foreach \x in {1,...,5} {
                \fill[black] (\x, -\x) rectangle (\x+1, -\x-1);
                }
            \foreach \x in {1,...,6} {
                \fill[black] (\x, 0) rectangle (\x+1, -1);
                \fill[black] (\x, -7) rectangle (\x+1, -8);
                }
        \end{tikzpicture}
                }
    \caption{An example of the \textsc{FrozenLake} environment with holes located in the diagonal.}
    \label{fig:environment-frozen-lake}
\end{figure}
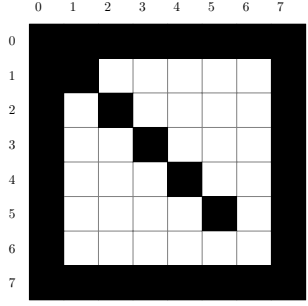
    
        Next, we will construct the oracle $\oracle_P^c$ that encodes the columns of $P$. Note that for this particular environment, any next state $s'$ arises as a transition from at most four state-action pairs corresponding to taking the action $a=s'-s$ from adjacent positions $s$. In other words, the value ${\rm c}_P$ can be chosen to be $4$ and we have: 
        $$ \oracle_P^c: \ket{0_s,0_a}\ket{s'} \longrightarrow \frac{1}{2}\sum_{\substack{s\in \mathcal{F} | s+a=s'}} \ket{s,a}\ket{s'} + \ket{G^\perp_{s'}}$$
        The mapping $\ket{0_s,0_a}\ket{s'}\xrightarrow[]{}\frac{1}{2}\sum_{s,a|s+a=s'}\ket{s,a}\ket{s'}$, that builds the superposition of the four possible ancestors of $s'$, can be implemented in linear cost. 
        It can be combined with $\oracle_\mathcal{F}$ to construct $\oracle_P^c$ with total cost $\bigO(\cost_\mathcal{F})$.
        
        Finally, we also need to build the oracle $\oracle_R$ that gives access to the reward vector $R$. For example, assuming that all four adjacent positions to the goal state $s_G$ are walkable, we can rewrite $\oracle_R$ as:
        $$ 
        \oracle_R: \ket{0_s,0_a}\longrightarrow\ket{R}= \frac{1}{2} \sum_{a} \ket{s_G-a,a}.
        $$
        Similar constructions are easy to design when there are some holes around the goal state $s_G$.
        The sum on the above oracles contain at most four elements and can be implemented with cost $\bigO(T_G)$ by first mapping $\ket{0_s}$ to $\ket{s_G}$ using $\oracle_G$ and then mapping $\ket{s_G,0_a}$ to $\ket{R}$.
    \end{proof}
    
    We have shown how to implement the necessary oracles that give quantum access to $\mathcal{M}$. It is important to note that we have set ${\rm c}_P=4$ which implies that the normalizing factor $\mu_P=2$ of the block-encoding of the transition matrix $P$ is fixed and does not depend on the grid size $S = X\times Y$. Given some policy $\pi$, the total cost of our quantum policy evaluation procedure from Theorem \ref{theorem:quantum-policy-evaluation} is then $\bigO(\Gamma(\cost_\mathcal{F}+\cost_G+\cost_\pi))$ where $\cost_\pi$ is the cost for the oracle encoding $\pi$ as in Definition \ref{definition:block-policy}. In the general case, we can implement $\oracle_\mathcal{F}$, $\oracle_G$ and $\oracle_\pi$ with $\bigO(XY)$ qubits and $\bigO(\polylog(XY))$ depth. However, there exists specific cases where the implementation of the oracle $\oracle_\mathcal{F}$ requires only $\bigO(\polylog(XY))$ qubits and depth. If for example all the non-walkable positions are located in the diagonal of the grid as in Figure \ref{fig:environment-frozen-lake}, we can implement $\oracle_\mathcal{F}$ using elementary mappings that require only $\bigO(\polylog(XY))$ qubits. 
    
    \indent \textbf{Running time:} The running time of quantum policy evaluation is $\bigO(\Gamma\-\polylog(XY))$ and the total running time of quantum policy iteration is $\bigO(M\Gamma\polylog(XY))$ with $M$ being the number of measurements. Executing the same classical algorithm yields a running time of $\bigO((XY)^\omega)$ since $A=4$ and does not depend on $XY$. Whether or not we have a quantum advantage depends on what is the required value of $M$. Setting $M=\bigO(\log(XY)/\epsilon^2)$ may not suffice when $XY$ is very large because the value function concentrates around the goal state and the $\ell_\infty$-tomography only guarantees $\epsilon$-approximation in average. In this case, we may use $\ell_2$-tomography with $M=\widetilde{\bigO}(XY/\epsilon^2)$ to guarantee that each grid position is sampled enough and the running time becomes $\widetilde{\bigO}(XY\Gamma/\epsilon^2)$ where the value of $\epsilon$ does not depend on the grid size, which still gives us a polynomial speedup over the classical in the worst case. 
    
    \indent \textbf{Experimental results:} We simulated the quantum policy iteration on a classical computer by introducing the appropriate noise and randomness within the linear algebraic procedures of the algorithm. More precisely, two types of noise were added to the normalized state-value function evaluated with a classical procedure. Given a precision parameter $\epsilon$, the first noise corresponds to the matrix inversion error (Theorem \ref{theorem:quantum-linear-algebra}) in the quantum policy evaluation method, whereas the second noise corresponds to the sampling error due to the finite number of quantum measurements (Theorem \ref{theorem:tomography}) where the number of measurements is chosen to be $M={36 \log(SA)}/{\epsilon^2}=36 \log(4XY)/\epsilon^2 $ as in \cite{Kerenidis2020QuantumAF}. We used $5$ different random seeds to run our experiments on the $4\times 4$ and $8 \times 8$ maps for the \textsc{FrozenLake} environment \cite{Brockman2016OpenAIG} and we saw that the quantum policy iteration converges to the optimal policy after at most five iterations for a precision parameter $\epsilon=10^{-2}$. 
    
\subsection{Application to \textsc{InvertedPendulum}}
\label{subsection:exp-InvertedPendulum}

    \indent \textbf{Description of the environment:} \textsc{InvertedPendulum} is an environment that requires maintaining a pendulum in a stable position by moving the cart it is attached to \cite{Wang1996AnAT}.  The space state $\mathcal{S}\subset \mathbb{R}^2$ is continuous and consists of tuples of the form $s=(\theta,\dot{\theta})$ where $\theta\in[-\pi/2,\pi/2]$ is the vertical angle and $\dot{\theta}\in\mathbb{R}$ the velocity. The action space consists of three Newtonian forces $\mathcal{A}=\{-50N,0N,+50N\}$ that can be applied to the cart to balance the pendulum. A uniform noise in $[-10,10]$ is added to any action. The game stops when the angle is greater than $\pi/2$ in absolute value. 
    
    The dynamics of the environment are governed by the following equation: $$\ddot{\theta} = \frac{g \sin(\theta) - \alpha ml \dot{\theta}^2\sin(2\theta)/2-\alpha \cos(\theta) a}{4l/3-\alpha ml\cos^2(\theta)}$$ where $g$ is the gravity constant, $m$ is the mass of the pendulum, $M$ is the mass of the cart, $l$ is the length of the pendulum and $\alpha = 1/(m+M)$.  
    
    \indent \textbf{Quantum access:} Since the state space $\mathcal{S}$ is continuous, we will apply the model-free implementation of quantum policy iteration. We want to build quantum access to a source $\mathcal{D}$ of transition samples classically collected from the \textsc{InvertedPendulum} environment and to some features function $\Phi$ by constructing the oracles as in Definition \ref{definition:block-model-free}. 
    
    First, let us consider quantum access to $\mathcal{D}$ by constructing $\oracle_{\widetilde{P}}^{s,a}$, $\oracle_{\widetilde{P}}^{s'}$. Assuming that we have a $B$-bit binary description for the states and actions, we can implement both oracles with cost $\cost_{\widetilde{P}}=\bigO(D\times B)$ where $D$ is the number of samples in $\mathcal{D}$. Moreover, we can also implement the oracle $\oracle_{\widetilde{R}}$ with cost $\cost_{\widetilde{R}}=\bigO(\polylog(D))$ since all rewards have value $+1$ and the approximated reward vector $\ket{\widetilde{R}}= \sum_i \ket{i}/\sqrt{D}$ can be implemented by applying a Hadamard transform to $\ket{0_i}$. 
    
    Next, we will construct an oracle that implements the features function $\Phi$. In particular, we will use the Fourier features \cite{Konidaris2011ValueFA}. Given some policy $\pi$, we will approximate the value function $Q^\pi$ using a multivariate Fourier series expansion of $Q^\pi(\cdot,a)$ on $[-1,1]^2$:
    $$Q^\pi(s,a) = \sum_{\mathbf{c}} \alpha_{\mathbf{c}} \cos(\pi \mathbf{c}\cdot s)+\beta_{\mathbf{c}} \sin(\pi \mathbf{c}\cdot s) \ \ \ \text{with} \ \ \mathbf{c}\in \mathbb{N}^{\text{dim}(\mathcal{S})}=\mathbb{N}^2$$
    
    To do so, we rescale the state parameters to range in $[0,1]^2$. We limit the expansion to some degree $k\in \mathbb{N}$ by considering coefficients $\mathbf{c}\in [k]^2=\{0,\dots,k-1\}^2$  which results in $2k^2$ features per action and a total number of $K=2Ak^{\text{dim}(\mathcal{S})}=6k^2$ features. The features function $\Phi: \mathcal{S}\times\mathcal{A} \longrightarrow \mathbb{R}^K$ maps every state-action pair $(s,a)$ to $\Phi(s,a)=\{\phi_{\cos}^{\mathbf{c},\mathbf{a}}(s,a),\phi_{\sin}^{\mathbf{c},\mathbf{a}}(s,a)|{\mathbf{c}\in[k]^2,\mathbf{a}\in\mathcal{A}}\}$ where: $$\phi_{\cos}^{\mathbf{c},\mathbf{a}}(s,a) = \mathds{1}[\mathbf{a}=a] \cos(\pi \mathbf{c}\cdot s) \quad\quad \text{and} \quad\quad \phi_{\sin}^{\mathbf{c},\mathbf{a}}(s,a) = \mathds{1}[\mathbf{a}=a] \sin(\pi \mathbf{c}\cdot s)$$ 
    
    In the next claim, we show how to efficiently implement the oracle $\oracle_\Phi$ associated to the features function $\Phi$:

    \begin{claim}\label{claim-cart-pole} 
         Let $k\in\mathbb{N}$ be the Fourier degree expansion and $B$ the number of bits used to describe $s$. We can implement the oracle $\oracle_\Phi$ with cost  $\cost_\Phi=\bigO(B\polylog(K))$.
    \end{claim}
    
    \begin{proof}
        We need to build quantum access to the oracle $\oracle_\Phi: \ket{s,a}\ket{0_k}\xrightarrow{} \ket{s,a}\ket{\Phi(s,a)}$. It is important to note that $\|\Phi(s,a)\|$ is constant for all state-action pairs and that the corresponding quantum state $\ket{\Phi(s,a)}$ can be written as: $$ \ket{\Phi(s,a)} = \frac{1}{k}\sum_{\mathbf{c}\in[k]^2} (\cos(\pi \mathbf{c}\cdot s)\ket{0}+\sin(\pi \mathbf{c}\cdot s)\ket{1})\ket{\mathbf{c},a}$$
        Note that the features register $\ket{0_k}$ can be decomposed into three registers $\ket{0}\ket{0_{\mathbf{c}},0_a}$ that index the $K=6k^2$ features. Moreover, the state representation $\ket{s}$ can also be decomposed into $\ket{\theta,\dot{\theta}}$ and we assume that every observation is encoded as a $B$-bit binary description. 
        
        Starting from $\ket{s,a}\ket{0_k}=\ket{s,a}\ket{0}\ket{0_{\mathbf{c}},0_a}$, we can map $\ket{0_{\mathbf{c}}}$ to $\sum_{\mathbf{c}\in[k]^2}\ket{\mathbf{c}}/k$ with cost $\bigO(\polylog(K))$. Then we use $B$-ancilla qubits to compute and store the results of the mapping $\ket{s}\ket{\mathbf{c}}\ket{0_B}\longrightarrow\ket{s}\ket{\mathbf{c}}\ket{\mathbf{c}\cdot s}$ with cost $\bigO(B\polylog(K))$ using quantum circuits for addition and multiplication. Next, we control on $\ket{\mathbf{c}\cdot s}$ to map the first qubit of the features register to $(\cos(\pi \mathbf{c}\cdot s)\ket{0}+\sin(\pi \mathbf{c}\cdot s)\ket{1})$ with cost $B$. Finally, we finish by uncomputing $\ket{\mathbf{c}\cdot s}$ and copy the action register of $\ket{s,a}$ to $\ket{0_a}$.
    \end{proof}
    
     We have shown how to implement the necessary oracles that give quantum access to the parameters of our model-free quantum approximate policy evaluation. Using Lemma \ref{lemma:block-model-free}, we can use these oracles to construct $\sqrt{K}$-block-encodings of the matrices $\widetilde{\Phi}$ and $\widetilde{P^\pi\Phi}$ with cost $\cost_{\widetilde{\Phi}}=\bigO(BD + B\polylog(K))$ where we assumed that the implementation of $\widetilde{\pi}$ has the same cost as the one of $\widetilde{P}$. The total cost of our evaluation procedure from Theorem \ref{theorem:model-free-evaluation} simplifies to  $\bigO(\kappa^2_{\widetilde{\Phi}}BD\Gamma\polylog(K\kappa_{\widetilde{\Phi}}/\epsilon))$. 

    \indent \textbf{Running time:} As demonstrated in the analysis above, the implementation of the oracles that give access to the memory $D$ require at most $\bigO(BD)$ qubits and can be performed with constant depth. However, the implementation of the features function uses quantum circuits with a linear dependency on $B$ since we need to control on the $B$ qubits used to store the values $\mathbf{c}\cdot s$. The running time of quantum approximate policy evaluation simplifies then to $\bigO(\kappa^2_{\widetilde{\Phi}}B\polylog(K\kappa_{\widetilde{\Phi}}/\epsilon))$. Since we need to improve the policy for every transition state $s$ in $\mathcal{D}$, the total running time of quantum approximate policy iteration is then $\widetilde{\bigO}(\kappa_{\widetilde{\Phi}|\mathcal{S}}\kappa^2_{\widetilde{\Phi}}MDB\polylog(K\kappa_{\widetilde{\Phi}}/\epsilon))$ where $M$ is the total number of measurements performed to update one state $s$. In comparison, executing classically this algorithm takes $\widetilde{\mathcal{O}}(BDK^2 + BK^w)$ for the approximate policy evaluation that computes $\widehat{w}^\pi=$ and  $\widetilde{\mathcal{O}}(BDK)$ for the approximate policy improvement step where we use the bit time complexity for matrix multiplication and inversion. Our analysis show that both approaches have linear dependency on $B$ and $D$, however our quantum algorithm provides a polynomial speedup in the total number of features $K=2Ak^{\text{dim}(\mathcal{S})}$ that grows exponentially with the dimension of the state space if we apply this approach to other environments. 
    
    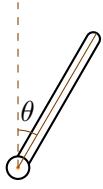
\begin{figure}[t]
    \centering
    
    \begin{tikzpicture}[thick]
     
        \newcommand{\ang}{30}
        \begin{scope} [draw = black,
            fill = white, 
            dot/.style = {orange, radius = .025}]
         \filldraw [rotate around = {-\ang:(0,1.5)}] (.09,1.5) -- 
            +(0,2) arc (0:180:.09) 
            coordinate [pos = .5] (T) -- (-.09,1.5);
        \filldraw (0,1.52) circle (.15);
        \fill[dot] (0,1.52) circle;
        \end{scope}
        \begin{scope} [thin, orange!50!black]
            \draw (T) -- (0,1.52) coordinate (P);
            \draw[dashed] (P) + (0,0) -- +(0,2.2);
            \draw (P) + (0,.5) arc (90:90-\ang:.5) node [black, midway, above] {$\theta$};
        \end{scope}
    \end{tikzpicture}
    \caption{\textsc{InvertedPendulum} environment.}
    \label{fig:cart-pole}
\end{figure}
    
    \indent \textbf{Experimental results:} We simulated the model-free implementation of quantum approximate policy iteration on a classical computer with a Fourier expansion of degree $k=4$ for a total of $K=96$ features per state-action pair. Similarly to the experiments in Subsection \ref{subsection:exp-frozen-lake}, we added a noise of magnitude $\epsilon$ to all algebraic procedures and we performed $M=100$ measurements for every sample in the memory. We preprocessed the state to range in $[0,1]^2$ by normalizing the angle and clipping the angle velocity $\dot{\theta}$ between $[-1,1]$ before rescaling to $[0,1]$. All other simulation parameters are identical to those in \cite{Lagoudakis2003LeastSquaresPI}. Moreover, we also clipped the singular values of $\mathbf{A}^\pi= \widetilde{\Phi}^\top (\widetilde{\Phi}-\gamma \widetilde{P^\pi\Phi})$ so that its corresponding condition number is constant and has value $\kappa=1/10^{-3}$. We repeated the experiment over $5$ different random seeds with a precision $\epsilon=10^{-2}$ and saw that our algorithm converges to the optimal policy within the first $8$ iterations. 
    
\section{Conclusion and discussions}
\label{section:conclusion}

    \indent In this work, we provided a general framework for performing quantum reinforcement learning via exact and approximate policy iteration. We validated our framework by designing and analyzing quantum policy evaluation methods for infinite horizon discounted problems by building quantum states that approximately encode the value function of a policy $\pi$, and quantum policy improvement methods by post-processing measurement outcomes on these quantum states. In all cases, we provided details about constructing block encodings for all matrices needed in the quantum linear algebra computations. Last, we studied the theoretical and experimental performance of our quantum algorithms on the \textsc{FrozenLake} and \textsc{InvertedPendulum} environments.
    
    \indent Our framework can be adapted and generalized to encompass many different policy iteration algorithms, including ones using deep learning techniques, and, of course, further theoretical work is needed in order to fully understand the strengths and limits of this approach. We conclude by providing several directions for possible future work.
    
    \indent First, the cost and running time of our quantum policy evaluation algorithms have linear dependency on the quantities $\mu$ and $\kappa$ of the different matrices appearing in the linear systems used to compute or estimate the value function. The condition number $\kappa$ is a property of the matrix and cannot be optimized, but one can use a much smaller threshold $\kappa_{th}$, thus disregarding smaller eigenvalues, a method that works well when there is a good low rank approximation of the matrix. The quantity $\mu$ depends on the procedure used to build quantum access as in the block-encoding framework and different methods will provide different $\mu$ parameters. We have provided examples where both these parameters are small, but it remains open to understand the families of environments for which quantum linear algebra can be faster than classical methods.
    
    \indent Second, we provided several quantum policy improvements strategies that consist of performing a series of measurements on the outputs of quantum policy evaluation to update the actual policy. Again, the running time of one iteration of our algorithm is linearly dependent on the number of measurements which also affects the overall performance of our policy. Moreover, there is also inherent noise induced from measurements that is specific to the quantum procedures. We have set this number, for most of our results, to be $\widetilde{\bigO}(1/\epsilon^2)$ for the theoretical guarantees provided by $\ell_\infty$-tomography but this number may be far from optimal. Possible research directions include adaptively controlling this number or/and making it state-dependent to appropriately balance between exploration-exploitation. If no exploration is needed, we can instead focus on finding the correct argmax using the quantum maximum finding algorithm by Dürr and Høyer \cite{durr1996quantum} similarly to the approach in \cite{Wang2021QuantumAF}. Understanding better how the number of measurements affects the convergence and performance of the quantum reinforcement learning methods needs to be more thoroughly explored.
    
    \indent Last, one may also study the different variants of classical policy iteration that exist and try to provide similar theoretical guarantees of convergence for some appropriate norm for the quantum case. 

\newpage
\bibliographystyle{unsrtnat}
\bibliography{bibliography}

\begin{thebibliography}{34}
\providecommand{\natexlab}[1]{#1}
\providecommand{\url}[1]{\texttt{#1}}
\expandafter\ifx\csname urlstyle\endcsname\relax
  \providecommand{\doi}[1]{doi: #1}\else
  \providecommand{\doi}{doi: \begingroup \urlstyle{rm}\Url}\fi

\bibitem[Mnih et~al.(2015)Mnih, Kavukcuoglu, Silver, Rusu, Veness, Bellemare,
  Graves, Riedmiller, Fidjeland, Ostrovski, Petersen, Beattie, Sadik,
  Antonoglou, King, Kumaran, Wierstra, Legg, and
  Hassabis]{Mnih2015HumanlevelCT}
Volodymyr Mnih, Koray Kavukcuoglu, David Silver, Andrei~A. Rusu, Joel Veness,
  Marc~G. Bellemare, Alex Graves, Martin~A. Riedmiller, Andreas Fidjeland,
  Georg Ostrovski, Stig Petersen, Charlie Beattie, Amir Sadik, Ioannis
  Antonoglou, Helen King, Dharshan Kumaran, Daan Wierstra, Shane Legg, and
  Demis Hassabis.
\newblock Human-level control through deep reinforcement learning.
\newblock \emph{Nature}, 518:\penalty0 529--533, 2015.

\bibitem[Silver et~al.(2017)Silver, Schrittwieser, Simonyan, Antonoglou, Huang,
  Guez, Hubert, baker, Lai, Bolton, Chen, Lillicrap, Hui, Sifre, van~den
  Driessche, Graepel, and Hassabis]{Silver2017MasteringTG}
David Silver, Julian Schrittwieser, Karen Simonyan, Ioannis Antonoglou, Aja
  Huang, Arthur Guez, Thomas Hubert, Lucas baker, Matthew Lai, Adrian Bolton,
  Yutian Chen, Timothy~P. Lillicrap, Fan Hui, L.~Sifre, George van~den
  Driessche, Thore Graepel, and Demis Hassabis.
\newblock Mastering the game of go without human knowledge.
\newblock \emph{Nature}, 550:\penalty0 354--359, 2017.

\bibitem[Szegedy et~al.(2014)Szegedy, Zaremba, Sutskever, Bruna, Erhan,
  Goodfellow, and Fergus]{Szegedy2014IntriguingPO}
Christian Szegedy, Wojciech Zaremba, Ilya Sutskever, Joan Bruna, D.~Erhan,
  Ian~J. Goodfellow, and Rob Fergus.
\newblock Intriguing properties of neural networks.
\newblock \emph{CoRR}, abs/1312.6199, 2014.

\bibitem[Arute et~al.(2019)Arute, Arya, Babbush, Bacon, Bardin, Barends,
  Biswas, Boixo, Brand{\~a}o, Buell, Burkett, Chen, Chen, Chiaro, Collins,
  Courtney, Dunsworth, Farhi, Foxen, Fowler, Gidney, Giustina, Graff, Guerin,
  Habegger, Harrigan, Hartmann, Ho, Hoffmann, Huang, Humble, Isakov, Jeffrey,
  Jiang, Kafri, Kechedzhi, Kelly, Klimov, Knysh, Korotkov, Kostritsa, Landhuis,
  Lindmark, Lucero, Lyakh, Mandr{\`a}, McClean, McEwen, Megrant, Mi,
  Michielsen, Mohseni, Mutus, Naaman, Neeley, Neill, Niu, Ostby, Petukhov,
  Platt, Quintana, Rieffel, Roushan, Rubin, Sank, Satzinger, Smelyanskiy, Sung,
  Trevithick, Vainsencher, Villalonga, White, Yao, Yeh, Zalcman, Neven, and
  Martinis]{Arute2019QuantumSU}
Frank Arute, Kunal Arya, Ryan Babbush, Dave Bacon, Joseph~C. Bardin, Rami
  Barends, Rupak Biswas, Sergio Boixo, Fernando G. S.~L. Brand{\~a}o, David~A.
  Buell, Brian Burkett, Yu~Chen, Zijun Chen, Benjamin Chiaro, Roberto Collins,
  William Courtney, Andrew Dunsworth, Edward Farhi, Brooks Foxen, Austin~G.
  Fowler, Craig Gidney, Marissa Giustina, Rob Graff, Keith Guerin, Steve
  Habegger, Matthew~P. Harrigan, Michael~J. Hartmann, Alan~K. Ho, Markus
  Hoffmann, Trent Huang, T.~Humble, Sergei~V. Isakov, Evan Jeffrey, Zhang
  Jiang, Dvir Kafri, Kostyantyn Kechedzhi, Julian Kelly, Paul Klimov, Sergey
  Knysh, Alexander~N. Korotkov, Fedor Kostritsa, David Landhuis, Mike Lindmark,
  Erik Lucero, Dmitry~I. Lyakh, Salvatore Mandr{\`a}, Jarrod~R. McClean,
  Matthew~J. McEwen, Anthony Megrant, Xiao Mi, Kristel Michielsen, Masoud
  Mohseni, Josh Mutus, Ofer Naaman, Matthew Neeley, Charles~J. Neill,
  Murphy~Yuezhen Niu, Eric~P. Ostby, Andre Petukhov, John~C. Platt, Chris
  Quintana, Eleanor~Gilbert Rieffel, Pedram Roushan, Nicholas~C Rubin,
  Daniel~Thomas Sank, Kevin~J Satzinger, Vadim~N. Smelyanskiy, Kevin~J. Sung,
  Matthew~D Trevithick, Amit Vainsencher, Benjamin Villalonga, Theodore White,
  Z.~Jamie Yao, P.~Yeh, Adam Zalcman, Hartmut Neven, and John~M. Martinis.
\newblock Quantum supremacy using a programmable superconducting processor.
\newblock \emph{Nature}, 574:\penalty0 505--510, 2019.

\bibitem[Lloyd et~al.(2014)Lloyd, Mohseni, and Rebentrost]{Lloyd2014QuantumPC}
Seth Lloyd, Masoud Mohseni, and Patrick Rebentrost.
\newblock Quantum principal component analysis.
\newblock \emph{Nature Physics}, 10:\penalty0 631--633, 2014.

\bibitem[Kerenidis and Prakash(2017)]{Kerenidis2017QuantumRS}
Iordanis Kerenidis and Anupam Prakash.
\newblock Quantum recommendation systems.
\newblock \emph{ArXiv}, abs/1603.08675, 2017.

\bibitem[Biamonte et~al.(2017)Biamonte, Wittek, Pancotti, Rebentrost, Wiebe,
  and Lloyd]{Biamonte2017QuantumML}
Jacob~D. Biamonte, Peter Wittek, Nicola Pancotti, Patrick Rebentrost, Nathan
  Wiebe, and Seth Lloyd.
\newblock Quantum machine learning.
\newblock \emph{Nature}, 549:\penalty0 195--202, 2017.

\bibitem[Lloyd and Weedbrook(2018)]{Lloyd2018QuantumGA}
Seth Lloyd and Christian Weedbrook.
\newblock Quantum generative adversarial learning.
\newblock \emph{Physical review letters}, 121 4:\penalty0 040502, 2018.

\bibitem[Kerenidis et~al.(2019)Kerenidis, Landman, and
  Prakash]{Kerenidis2020QuantumAF}
Iordanis Kerenidis, Jonas Landman, and Anupam Prakash.
\newblock Quantum algorithms for deep convolutional neural networks.
\newblock In \emph{International Conference on Learning Representations}, 2019.

\bibitem[Kerenidis and Prakash(2020)]{Kerenidis2020QuantumGD}
Iordanis Kerenidis and Anupam Prakash.
\newblock Quantum gradient descent for linear systems and least squares.
\newblock \emph{Physical Review A}, 101:\penalty0 022316, 2020.

\bibitem[Kerenidis et~al.(2021)Kerenidis, Landman, and
  Mathur]{Kerenidis2021ClassicalAQ}
Iordanis Kerenidis, Jonas Landman, and Natansh Mathur.
\newblock Classical and quantum algorithms for orthogonal neural networks.
\newblock \emph{ArXiv}, abs/2106.07198, 2021.

\bibitem[Sutton and Barto(2005)]{Sutton2005ReinforcementLA}
Richard~S. Sutton and Andrew~G. Barto.
\newblock Reinforcement learning: An introduction.
\newblock \emph{IEEE Transactions on Neural Networks}, 16:\penalty0 285--286,
  2005.

\bibitem[Dong et~al.(2008)Dong, Chen, Li, and Tarn]{Dong2008QuantumRL}
Daoyi Dong, Chunlin Chen, Hanxiong Li, and Tzyh-Jong Tarn.
\newblock Quantum reinforcement learning.
\newblock \emph{IEEE Transactions on Systems, Man, and Cybernetics, Part B
  (Cybernetics)}, 38\penalty0 (5):\penalty0 1207--1220, 2008.

\bibitem[Grover(1996)]{Grover1996AFQ}
Lov~K. Grover.
\newblock A fast quantum mechanical algorithm for database search.
\newblock In \emph{Proceedings of the twenty-eighth annual ACM symposium on
  Theory of computing}, pages 212--219, 1996.

\bibitem[Cornelissen(2018)]{Cornelissen2018QuantumGE}
Arjan Cornelissen.
\newblock Quantum gradient estimation and its application to quantum
  reinforcement learning.
\newblock Master's thesis, Delft University of Technology, 2018.

\bibitem[Gily{\'e}n et~al.(2019{\natexlab{a}})Gily{\'e}n, Arunachalam, and
  Wiebe]{gilyen2019optimizing}
Andr{\'a}s Gily{\'e}n, Srinivasan Arunachalam, and Nathan Wiebe.
\newblock Optimizing quantum optimization algorithms via faster quantum
  gradient computation.
\newblock In \emph{Proceedings of the Thirtieth Annual ACM-SIAM Symposium on
  Discrete Algorithms}, pages 1425--1444. SIAM, 2019{\natexlab{a}}.

\bibitem[Ronagh(2019)]{Ronagh2019QuantumAF}
Pooya Ronagh.
\newblock Quantum algorithms for solving dynamic programming problems.
\newblock \emph{ArXiv}, abs/1906.02229, 2019.

\bibitem[Chen et~al.(2020)Chen, Yang, Qi, Chen, Ma, and
  Goan]{Chen2020VariationalQC}
Samuel Yen-Chi Chen, Chao-Han~Huck Yang, Jun Qi, Pin-Yu Chen, Xiaoli Ma, and
  Hsi-Sheng Goan.
\newblock Variational quantum circuits for deep reinforcement learning.
\newblock \emph{IEEE Access}, 8:\penalty0 141007--141024, 2020.

\bibitem[Lockwood and Si(2020)]{Lockwood2020ReinforcementLW}
Owen Lockwood and M.~Si.
\newblock Reinforcement learning with quantum variational circuits.
\newblock \emph{ArXiv}, abs/2008.07524, 2020.

\bibitem[Skolik et~al.(2021)Skolik, Jerbi, and Dunjko]{Skolik2021QuantumAI}
Andrea Skolik, Sofi{\`e}ne Jerbi, and Vedran Dunjko.
\newblock Quantum agents in the gym: a variational quantum algorithm for deep
  q-learning.
\newblock \emph{ArXiv}, abs/2103.15084, 2021.

\bibitem[Jerbi et~al.(2021)Jerbi, Gyurik, Marshall, Briegel, and
  Dunjko]{Jerbi2021VariationalQP}
Sofi{\`e}ne Jerbi, Casper Gyurik, Simon Marshall, Hans~J. Briegel, and Vedran
  Dunjko.
\newblock Variational quantum policies for reinforcement learning.
\newblock \emph{ArXiv}, abs/2103.05577, 2021.

\bibitem[Wang et~al.(2021)Wang, Sundaram, Kothari, Kapoor, and
  R{\"o}tteler]{Wang2021QuantumAF}
Daochen Wang, Aarthi Sundaram, Robin Kothari, Ashish Kapoor, and Martin
  R{\"o}tteler.
\newblock Quantum algorithms for reinforcement learning with a generative
  model.
\newblock In \emph{ICML}, 2021.

\bibitem[Lagoudakis and Parr(2003)]{Lagoudakis2003LeastSquaresPI}
Michail~G. Lagoudakis and Ronald~E. Parr.
\newblock Least-squares policy iteration.
\newblock \emph{J. Mach. Learn. Res.}, 4:\penalty0 1107--1149, 2003.

\bibitem[Bertsekas(2011)]{Bertsekas2011ApproximatePI}
Dimitri~P. Bertsekas.
\newblock Approximate policy iteration: a survey and some new methods.
\newblock \emph{Journal of Control Theory and Applications}, 9:\penalty0
  310--335, 2011.

\bibitem[Scherrer et~al.(2012)Scherrer, Gabillon, Ghavamzadeh, and
  Geist]{Scherrer2012ApproximateMP}
Bruno Scherrer, Victor Gabillon, Mohammad Ghavamzadeh, and Matthieu Geist.
\newblock Approximate modified policy iteration.
\newblock \emph{ArXiv}, abs/1205.3054, 2012.

\bibitem[Bertsekas(2019)]{Bertsekas2019Reinforcement}
Dimitri~P. Bertsekas.
\newblock \emph{Reinforcement learning and optimal control}.
\newblock Athena Scientific, 2019.

\bibitem[Nielsen and Chuang(2002)]{Nielsen2000QuantumCA}
Michael~A Nielsen and Isaac Chuang.
\newblock Quantum computation and quantum information, 2002.

\bibitem[Chakraborty et~al.(2018)Chakraborty, Gily{\'e}n, and
  Jeffery]{Chakraborty2019ThePO}
Shantanav Chakraborty, Andr{\'a}s Gily{\'e}n, and Stacey Jeffery.
\newblock The power of block-encoded matrix powers: improved regression
  techniques via faster hamiltonian simulation.
\newblock \emph{ArXiv}, abs/1804.01973, 2018.

\bibitem[Gily{\'e}n et~al.(2019{\natexlab{b}})Gily{\'e}n, Su, Low, and
  Wiebe]{Gilyn2019QuantumSV}
Andr{\'a}s Gily{\'e}n, Yuan Su, Guang~Hao Low, and Nathan Wiebe.
\newblock Quantum singular value transformation and beyond: exponential
  improvements for quantum matrix arithmetics.
\newblock \emph{Proceedings of the 51st Annual ACM SIGACT Symposium on Theory
  of Computing}, 2019{\natexlab{b}}.

\bibitem[Munos(2003)]{Munos2003ErrorBF}
R{\'e}mi Munos.
\newblock Error bounds for approximate policy iteration.
\newblock In \emph{ICML}, volume~3, pages 560--567, 2003.

\bibitem[Brockman et~al.(2016)Brockman, Cheung, Pettersson, Schneider,
  Schulman, Tang, and Zaremba]{Brockman2016OpenAIG}
Greg Brockman, Vicki Cheung, Ludwig Pettersson, Jonas Schneider, John Schulman,
  Jie Tang, and Wojciech Zaremba.
\newblock Openai gym.
\newblock \emph{ArXiv}, abs/1606.01540, 2016.

\bibitem[Wang et~al.(1996)Wang, Tanaka, and Griffin]{Wang1996AnAT}
Hua~O. Wang, Kazuo Tanaka, and Michael~F. Griffin.
\newblock An approach to fuzzy control of nonlinear systems: stability and
  design issues.
\newblock \emph{IEEE Trans. Fuzzy Syst.}, 4:\penalty0 14--23, 1996.

\bibitem[Konidaris et~al.(2011)Konidaris, Osentoski, and
  Thomas]{Konidaris2011ValueFA}
George Konidaris, Sarah Osentoski, and Philip Thomas.
\newblock Value function approximation in reinforcement learning using the
  fourier basis.
\newblock In \emph{Twenty-fifth AAAI conference on artificial intelligence},
  2011.

\bibitem[Dürr and Høyer(1996)]{durr1996quantum}
Christoph Dürr and Peter Høyer.
\newblock A quantum algorithm for finding the minimum.
\newblock \emph{arXiv preprint quant-ph/9607014}, 1996.

\end{thebibliography}

\end{document}